\documentclass[11pt]{article}

\usepackage{enumitem} % Provides the ability to easily modify spacings related to lists
\setlist{noitemsep,topsep=0pt,parsep=0pt,partopsep=0pt,listparindent=\parindent}

\usepackage[margin=1cm]{caption} % Adds an addition margin on either side of figure captions. This must come before "\usepackage{subcaption}" (otherwise there are errors).
\usepackage{subcaption}
\usepackage{tikz}
\usetikzlibrary{arrows,backgrounds,calc,fit,shapes.geometric}

\tikzset{every fit/.append style=text badly centered}

\usepackage{tyson}

\usepackage[bookmarks=true,hypertexnames=false]{hyperref}

\newcommand{\Holant}{\operatorname{Holant}}
\newcommand{\PlHolant}{\operatorname{Pl-Holant}}
\newcommand{\holant}[2]{\Holant(#1\mid #2)}
\newcommand{\plholant}[2]{\PlHolant(#1\mid #2)}
\newcommand{\CSP}{\operatorname{\#CSP}}
\newcommand{\PlCSP}{\operatorname{Pl-\#CSP}}
\newcommand{\arity}{\operatorname{arity}}

\newcommand{\EQ}{\mathcal{EQ}}

\def\borderColor{blue!60}
\def\scale{0.6}
\def\nodeDist{1.4cm}
\tikzstyle{internal} = [draw, fill, shape=circle]
\tikzstyle{external} = [shape=circle]
\tikzstyle{square}   = [draw, fill, rectangle]
\tikzstyle{triangle} = [draw, fill, regular polygon, regular polygon sides=3, inner sep=2.5pt] % requires \usetikzlibrary{shapes.geometric}

\title{The Complexity of Planar Boolean \#CSP with Complex Weights}

\author{
 Heng Guo\\
 \scriptsize University of Wisconsin-Madison\\
 \scriptsize \texttt{hguo@cs.wisc.edu}
 \and
 Tyson Williams\\
 \scriptsize University of Wisconsin-Madison\\
 \scriptsize \texttt{tdw@cs.wisc.edu}
}

\date{} % no date

\begin{document}
\maketitle

\begin{abstract}
 We prove a complexity dichotomy theorem for symmetric complex-weighted Boolean \#CSP when the constraint graph of the input must be planar.
 The problems that are \#P-hard over general graphs but tractable over planar graphs are precisely those with a holographic reduction to matchgates.
 This generalizes a theorem of Cai, Lu, and Xia for the case of real weights.
 We also obtain a dichotomy theorem for a symmetric arity~4 signature with complex weights in the planar Holant framework,
 which we use in the proof of our \#CSP dichotomy.
 In particular, we reduce the problem of evaluating the Tutte polynomial of a planar graph at the point $(3,3)$
 to counting the number of Eulerian orientations over planar 4-regular graphs to show the latter is \#P-hard.
 This strengthens a theorem by Huang and Lu to the planar setting.
 Our proof techniques combine new ideas with refinements and extensions of existing techniques.
 These include planar pairings, the recursive unary construction, the anti-gadget technique, and pinning in the Hadamard basis.
\end{abstract}

\section{Introduction}

In 1979, Valiant~\cite{Val79b} defined the class $\SHARPP$ to explain the apparent intractability of counting the number of perfect matchings in a graph.
Yet over a decade earlier, Kasteleyn~\cite{Kas67} gave a polynomial-time algorithm to compute this quantity for planar graphs.
This was an important milestone in a decades-long research program by physicists in statistical mechanics
to determine what problems the restriction to the planar setting renders tractable~\cite{Isi25, Ons44, Yan52, YL52, LY52, TF61, Kas61, Kas67, Bax82, LS81, Wel93}.
More recently, Valiant introduced matchgates~\cite{Val02a, Val02b} and \emph{holographic} algorithms~\cite{Val08, Val06}
that rely on Kasteleyn's algorithm to solve certain counting problems over planar graphs.
In a series of papers~\cite{CC07b, CCL09, CL10a, CL11a},
Cai et al.~characterized the local constraint functions (which define counting problems) that are representable by matchgates in a holographic algorithm.

From the viewpoint of computational complexity, we seek to understand exactly which intractable problems the planarity restriction enable us to efficiently compute.
Partial answers to this question have been given in the context of various counting frameworks~\cite{Ver05, CLX10, CK12, CKW12}.
In every case, the problems that are $\SHARPP$-hard over general graphs but tractable over planar graphs are essentially those characterized by Cai et al.
In this paper, we give more evidence for this phenomenon by extending the results of \cite{CLX10} to the setting of complex-valued constraint functions.
This provides the most natural setting to express holographic algorithms and transformations.

Our main result is a dichotomy theorem for the framework of counting Constraint Satisfaction Problems (\#CSP),
but our proof is in a generalized framework called Holant problems~\cite{CLX11c, CLX12, CLX09a, CLX11d}.
We briefly introduce the Holant framework and then explain its main advantages.
A set of functions $\mathcal{F}$ defines the problem $\Holant(\mathcal{F})$.
An instance of this problem is a tuple $\Omega = (G, \mathcal{F}, \pi)$ called a \emph{signature grid},
where $G = (V,E)$ is a graph,
$\pi$ labels each $v \in V$ with a function $f_v \in \mathcal{F}$,
and $f_v$ maps $\{0,1\}^{\deg(v)}$ to $\mathbb{C}$.
We also call the functions in $\mathcal{F}$ \emph{signatures}.
An assignment $\sigma$ for every $e \in E$ gives an evaluation $\prod_{v \in V} f_v(\sigma \mid_{E(v)})$,
where $E(v)$ denotes the incident edges of $v$ and $\sigma \mid_{E(v)}$ denotes the restriction of $\sigma$ to $E(v)$.
The counting problem on the instance $\Omega$ is to compute
\[\Holant_\Omega = \sum_{\sigma : E \to \{0,1\}} \prod_{v \in V} f_v\left(\sigma \mid_{E(v)}\right).\]
Counting the number of perfect matchings in $G$ corresponds to attaching the \textsc{Exact-One} signature at every vertex of $G$.
A function or signature is called \emph{symmetric} if its output depends only on the Hamming weight of the input.
We often denote a symmetric signature by the list of its outputs sorted by input Hamming weight in ascending order.
For example, $[0,1,0,0]$ is the \textsc{Exact-One} function on three bits.
The output is~$1$ if and only if the input is $001$, $010$, or $100$, and~$0$ otherwise.

We consider \#CSP, which are also parametrized by a set of functions $\mathcal{F}$.
The problem $\CSP(\mathcal{F})$ is equivalent to $\Holant(\mathcal{F} \union \EQ)$,
where $\EQ = \{=_1, =_2, =_3, \dotsc\}$ and $(=_k) = [1, 0, \dotsc, 0, 1]$ is the equality signature of arity $k$.
This explicit role of equality signatures permits a finer classification of problems.
For a direct definition of \#CSP, see~\cite{DGJ09}.

We often consider a Holant problem over bipartite graphs, which is denoted by $\holant{\mathcal{F}}{\mathcal{G}}$,
where the sets $\mathcal{F}$ and $\mathcal{G}$ contain the signatures available for assignment to the vertices in each partition.
Considering the edge-vertex incidence graph, one can see that $\Holant(\mathcal{F})$ is equivalent to $\holant{=_2}{\mathcal{F}}$.
One powerful tool in this setting is the holographic transformation.
Let $T$ be a nonsingular $2$-by-$2$ matrix and define $T \mathcal{F} = \{T^{\otimes \arity(f)} f \st f \in \mathcal{F}\}$,
where $T^{\otimes k}$ is the tensor product of $k$ factors of $T$.
Here we view $f$ as a column vector by listing its values in lexicographical order as in a truth table.
Similarly $\mathcal{F} T$ is defined by viewing $f \in \mathcal{F}$ as a row vector.
Valiant's Holant theorem~\cite{Val08} states that $\holant{\mathcal{F}}{\mathcal{G}}$ is equivalent to $\holant{\mathcal{F} T^{-1}}{T \mathcal{G}}$.

Cai, Lu, and Xia gave a dichotomy for complex-weighted Boolean $\CSP(\mathcal{F})$ in~\cite{CLX09a}.
Let $\PlCSP(\mathcal{F})$ (resp.~$\PlHolant(\mathcal{F})$) denote the \#CSP (resp.~Holant problem)
defined by $\mathcal{F}$ when the inputs are restricted to planar graphs.
In this paper, we investigate the complexity of $\PlCSP(\mathcal{F})$ for a set $\mathcal{F}$ of symmetric complex-weighted functions.
In particular, we would like to determine which sets become tractable under this planarity restriction.
Holographic algorithms with matchgates provide planar tractable problems for sets that are matchgate realizable after a holographic transformation.
From the Holant perspective, the signatures in $\EQ$ are always available in $\CSP(\mathcal{F})$.
By the signature theory of Cai and Lu~\cite{CL11a},
the Hadamard matrix $H = \left[\begin{smallmatrix} 1 & 1 \\ 1 & -1 \end{smallmatrix}\right]$ essentially defines the
only\footnote{Up to transformations under which matchgates are closed.} holographic transformation under which $\EQ$ becomes matchgate realizable.
Let $\widehat{\mathcal{F}}$ denote $H \mathcal{F}$ for any set of signatures $\mathcal{F}$.
Then $\widehat{\EQ}$ is $\{[1,0], [1,0,1], [1,0,1,0], \dotsc\}$ while $(=_2) (H^{-1})^{\otimes 2}$ is still $=_2$.
Therefore $\CSP(\mathcal{F})$ and $\Holant(\mathcal{F} \union \EQ)$ are equivalent to $\Holant(\widehat{\mathcal{F}} \union \widehat{\EQ})$ by Valiant's Holant theorem.

Our main dichotomy theorem is stated as follows.
\begin{theorem} \label{thm:PlCSP:words}
 Let $\mathcal{F}$ be any set of symmetric, complex-valued signatures in Boolean variables.
 Then $\PlCSP(\mathcal{F})$ is $\SHARPP$-hard unless $\mathcal{F}$ satisfies one of the following conditions,
 in which case it is tractable:
 \begin{enumerate}
  \item $\CSP(\mathcal{F})$ is tractable (cf.~\cite{CLX09a}); or
  \item $\widehat{\mathcal{F}}$ is realizable by matchgates (cf.~\cite{CL11a}).
 \end{enumerate}
\end{theorem}
\noindent
A more explicit description of the tractable cases can be found in Theorem~\ref{thm:PlCSP}.

In many previous dichotomy theorems for Boolean $\CSP(\mathcal{F})$,
the proof of hardness began by pinning.
The goal of this technique is to realize the constant functions $[1,0]$ and $[0,1]$ and was always achieved by a \emph{nonplanar} reduction.
In the nonplanar setting, $[1,0]$ and $[0,1]$ are contained in each of the maximal tractable sets.
Therefore, pinning in this setting does not imply the collapse of any complexity classes.
However, $\EQ$ with $\{[1,0], [0,1]\}$ are not simultaneously realizable as matchgates.
If we are to prove our main theorem, one should not expect to be able to pin for $\PlCSP(\mathcal{F})$,
since otherwise $\SHARPP$ collapses to $\P$!
Instead, apply the Hadamard transformation and consider $\PlHolant(\widehat{\mathcal{F}} \union \widehat{\EQ})$.
In this Hadamard basis, pinning becomes possible again since $[1,0]$ and $[0,1]$ are included in every maximal tractable set.
Indeed, we prove our pinning result in this Hadamard basis in Section~\ref{sec:pinning}.

For Holant problems, it is often important to understand the complexity of the small arity cases first~\cite{CLX10, HL12, CGW12}.
In~\cite{CLX10}, Cai, Lu, and Xia gave a dichotomy for $\PlHolant(f)$ when $f$ is a symmetric arity~3 signature
while a dichotomy for $\Holant(f)$ when $f$ is a symmetric arity~4 signature was shown in~\cite{CGW12}.
In the proof of the latter result, most of the reductions were planar.
However, the crucial starting point for hardness,
namely counting Eulerian orientations (\#\textsc{EO}) over 4-regular graphs,
was not known to be $\SHARPP$-hard under the planarity restriction.
Huang and Lu~\cite{HL12} had recently proved that \#\textsc{EO} is \SHARPP-hard over 4-regular graphs but left open its complexity when the input is also planar.
We show that \#\textsc{EO} remains \SHARPP-hard over planar 4-regular graphs.
The problem we reduce from is the evaluation of the Tutte polynomial of a planar graph at the point (3,3),
which has a natural expression in the Holant framework.
In addition, we determine the complexity of counting complex-weighted matchings over planar 4-regular graphs.
The problem is $\SHARPP$-hard except for the tractable case of counting perfect matchings.
With these two ingredients, we obtain a dichotomy for $\PlHolant(f)$ when $f$ is a symmetric arity~4 signature.

Our main result is a generalization of the dichotomy by Cai, Lu, and Xia~\cite{CLX10} for $\PlCSP(\mathcal{F})$
when $\mathcal{F}$ contains symmetric real-weighted Boolean functions.
It is natural to consider complex weights in the Holant framework
because surprising equivalences between problems are often discovered via complex holographic transformations,
sometimes even between problems using only rational weights.
Our proof of hardness for \#EO over planar 4-regular graphs in Section~\ref{sec:EO} is a prime example of this.
Extending the range from $\mathbb{R}$ to $\mathbb{C}$ also enlarges the set of problems that can be transformed into the framework.

However, a dichotomy for complex weights is more technically challenging.
The proof technique of polynomial interpolation often has infinitely many failure cases in $\mathbb{C}$ corresponding to the infinitely many roots of unity,
which prevents a brute force analysis of failure cases as was done in~\cite{CLX10}.
This increased difficulty requires us to develop new ideas to bypass previous interpolation proofs.
In particular,
we perform a planar interpolation with a rotationally invariant signature to prove the $\SHARPP$-hardness of \#\textsc{EO} over planar 4-regular graphs.
For the complexity of counting complex-weighted matchings over planar 4-regular graphs,
we introduce the notion of planar pairings to build reductions.
We show that every planar 3-regular graph has a planar pairing and that it can be efficiently computed.
We also refine and extend existing techniques for application in the new setting,
including the recursive unary construction, the anti-gadget technique, compressed matrix criteria, and domain pairing.

This paper is organized as follows.
In Section~\ref{sec:preliminaries}, we give a review of terminology and previous dichotomy theorems.
In Section~\ref{sec:EO}, we prove that counting Eulerian orientations is $\SHARPP$-hard for planar 4-regular graphs.
In Section~\ref{sec:interpolation}, we strengthen a popular interpolation technique that uses recursive constructions, which leads to simpler proofs.
In Section~\ref{sec:arity-4}, we obtain our dichotomy theorem for $\PlHolant(f)$ when $f$ is a symmetric arity~4 signature with complex weights.
In Section~\ref{sec:pairing},
we prove several useful lemmas about a technique we call \emph{domain pairing} that essentially realizes an odd arity signature using only signatures of even arity.
In Section~\ref{sec:mixing}, we show that the three known sets of tractable signatures become $\SHARPP$-hard when mixed.
In Section~\ref{sec:pinning}, we use the pinning technique in a new planar proof to realize the constant functions $[1,0]$ and $[0,1]$.
In Section~\ref{sec:dichotomy}, we obtain our dichotomy theorem for $\PlCSP(\mathcal{F})$.

\section{Preliminaries} \label{sec:preliminaries}

\subsection{Problems and Definitions}

The framework of Holant problems is defined for functions mapping any $[q]^k \to \mathbb{F}$ for a finite $q$ and some field $\mathbb{F}$.
In this paper, we investigate complex-weighted Boolean $\Holant$ problems, that is, all functions are $[2]^k \to \mathbb{C}$.
Strictly speaking, for consideration of computational models, functions take complex algebraic numbers.

A \emph{signature grid} $\Omega = (G, \mathcal{F}, \pi)$ consists of a graph $G = (V,E)$,
where each vertex is labeled by a function $f_v \in \mathcal{F}$, and $\pi : V \to \mathcal{F}$ is the labeling.
If the graph $G$ is planar, then we call $\Omega$ a \emph{planar signature grid}.
The Holant problem on instance $\Omega$ is to evaluate $\Holant_\Omega = \sum_{\sigma} \prod_{v \in V} f_v(\sigma \mid_{E(v)})$,
a sum over all edge assignments $\sigma: E \to \{0,1\}$.

A function $f_v$ can be represented by listing its values in lexicographical order as in a truth table,
which is a vector in $\mathbb{C}^{2^{\deg(v)}}$,
or as a tensor in $(\mathbb{C}^{2})^{\otimes \deg(v)}$.
We also use $f^\alpha$ to denote the value $f(\alpha)$, where $\alpha$ is a binary string.
A function $f \in \mathcal{F}$ is also called a \emph{signature}.
A symmetric signature $f$ on $k$ Boolean variables can be expressed as $[f_0,f_1,\dotsc,f_k]$,
where $f_w$ is the value of $f$ on inputs of Hamming weight $w$.
In this paper, we consider symmetric signatures.
Sometimes we represent a signature of arity $k$ by a labeled vertex with $k$ ordered dangling edges corresponding to its input.

A Holant problem is parametrized by a set of signatures.

\begin{definition}
 Given a set of signatures $\mathcal{F}$,
 we define the counting problem $\Holant(\mathcal{F})$ as:

 Input: A \emph{signature grid} $\Omega = (G, \mathcal{F}, \pi)$;

 Output: $\Holant_\Omega$.
\end{definition}

The problem $\PlHolant(\mathcal{F})$ is defined similarly using a planar signature grid.
The $\Holant^c$ framework is the special case of Holant problems when the constant signatures of the domain are freely available.
In the Boolean domain, the constant signatures are $[1,0]$ and $[0,1]$.

\begin{definition}
 Given a set of signatures $\mathcal{F}$,
 $\Holant^c(\mathcal{F})$ denotes $\Holant(\mathcal{F} \cup \{[0,1],[1,0]\})$.
\end{definition}

The problem $\PlHolant^c(\mathcal{F})$ is defined similarly.
A signature $f$ of arity $n$ is \emph{degenerate} if there exist unary signatures $u_j \in \mathbb{C}^2$ ($1 \le j \le n$)
such that $f = u_1 \otimes \cdots \otimes u_n$.
A symmetric degenerate signature has the from $u^{\otimes n}$.
For such signatures, it is equivalent to replace it by $n$ copies of the corresponding unary signature.
Replacing a signature $f \in \mathcal{F}$ by a constant multiple $c f$,
where $c \ne 0$, does not change the complexity of $\Holant(\mathcal{F})$.
It introduces a global factor to $\Holant_\Omega$.
Hence, for two signatures $f,g$ of the same arity,
we use $f \neq g$ to mean that these signatures are not equal in the projective space sense,
i.e.~not equal up to any nonzero constant multiple.

We say a signature set $\mathcal{F}$ is tractable (resp.~$\SHARPP$-hard)
if the corresponding counting problem $\PlCSP(\mathcal{F})$ is tractable (resp.~$\SHARPP$-hard).
Sometimes we abuse this notation and also say that $\mathcal{F}$ is tractable to mean $\PlHolant(\mathcal{F})$ is tractable.
The intended counting problem should be clear from context.
Similarly for a signature $f$, we say $f$ is tractable (resp.~$\SHARPP$-hard) if $\{f\}$ is.
We follow the usual conventions about polynomial-time Turing reduction $\le_T$ and polynomial-time Turing equivalence $\equiv_T$.

\subsection{Holographic Reduction}

To introduce the idea of holographic reductions, it is convenient to consider bipartite graphs.
For a general graph, we can always transform it into a bipartite graph while preserving the Holant value, as follows.
For each edge in the graph, we replace it by a path of length two.
(This operation is called the \emph{2-stretch} of the graph and yields the edge-vertex incidence graph.)
Each new vertex is assigned the binary \textsc{Equality} signature $(=_2) = [1,0,1]$.

We use $\holant{\mathcal{F}}{\mathcal{G}}$ to denote the Holant problem on bipartite graphs $H = (U,V,E)$,
where each signature for a vertex in $U$ or $V$ is from $\mathcal{F}$ or $\mathcal{G}$, respectively.
An input instance for this bipartite Holant problem is a bipartite signature grid and is denoted by $\Omega = (H;\ \mathcal{F} \mid \mathcal{G};\ \pi)$.
Signatures in $\mathcal{F}$ are considered as row vectors (or covariant tensors);
signatures in $\mathcal{G}$ are considered as column vectors (or contravariant tensors)~\cite{DP91}.
Similarly, $\plholant{\mathcal{F}}{\mathcal{G}}$ denotes the Holant problem on planar bipartite graphs.

For a 2-by-2 matrix $T$ and a signature set $\mathcal{F}$,
define $T \mathcal{F} = \{g \mid \exists f \in \mathcal{F}$ of arity $n,~g = T^{\otimes n} f\}$, similarly for $\mathcal{F} T$.
Whenever we write $T^{\otimes n} f$ or $T \mathcal{F}$,
we view the signatures as column vectors;
similarly for $f T^{\otimes n} $ or $\mathcal{F} T$ as row vectors.
In the special case that $T = \left[\begin{smallmatrix} 1 & 1 \\ 1 & -1 \end{smallmatrix}\right]$,
we also define $T \mathcal{F} = \widehat{\mathcal{F}}$.

Let $T$ be an invertible 2-by-2 matrix.
The holographic transformation defined by $T$ is the following operation:
given a signature grid $\Omega = (H;\ \mathcal{F} \mid \mathcal{G};\ \pi)$, for the same graph $H$,
we get a new grid $\Omega' = (H;\ \mathcal{F} T \mid T^{-1} \mathcal{G};\ \pi')$ by replacing each signature in
$\mathcal{F}$ or $\mathcal{G}$ with the corresponding signature in $\mathcal{F} T$ or $T^{-1} \mathcal{G}$.

\begin{theorem}[Valiant's Holant Theorem~\cite{Val08}]
 If there is a holographic transformation mapping signature grid $\Omega$ to $\Omega'$,
 then $\Holant_\Omega = \Holant_{\Omega'}$.
\end{theorem}

Therefore, an invertible holographic transformation does not change the complexity of the Holant problem in the bipartite setting.
Furthermore, there is a special kind of holographic transformation, the orthogonal transformation,
that preserves the binary equality and thus can be used freely in the standard setting.

\begin{theorem}[Theorem~2.2 in~\cite{CLX09a}] \label{thm:orthogonal}
 Suppose $T$ is a 2-by-2 orthogonal matrix $(T \transpose{T} = I_2)$ and let $\Omega = (H, \mathcal{F}, \pi)$ be a signature grid.
 Under a holographic transformation by $T$, we get a new grid $\Omega' = (H, T \mathcal F, \pi')$ and $\Holant_\Omega = \Holant_{\Omega'}$.
\end{theorem}

Since the complexity of signatures are equivalent up to a nonzero constant factor,
we also call a transformation $T$ such that $T \transpose{T} = \lambda I$ for some $\lambda \neq 0$ an orthogonal transformation.
Such transformations do not change the complexity of a problem.

\subsection{Realization}

One basic notion used throughout the paper is realization.
We say a signature $f$ is \emph{realizable} or \emph{constructible} from a signature set $\mathcal{F}$
if there is a gadget with some dangling edges such that each vertex is assigned a signature from $\mathcal{F}$,
and the resulting graph, when viewed as a black-box signature with inputs on the dangling edges, is exactly $f$.
If $f$ is realizable from a set $\mathcal{F}$, then we can freely add $f$ into $\mathcal{F}$ preserving the complexity.

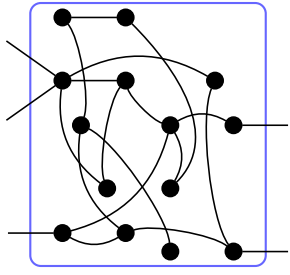
\begin{figure}[t]
 \centering
 \begin{tikzpicture}[scale=\scale,transform shape,node distance=\nodeDist,semithick]
  \node[external]  (0)                     {};
  \node[internal]  (1) [below right of=0]  {};
  \node[external]  (2) [below left  of=1]  {};
  \node[internal]  (3) [above       of=1]  {};
  \node[internal]  (4) [right       of=3]  {};
  \node[internal]  (5) [below       of=4]  {};
  \node[internal]  (6) [below right of=5]  {};
  \node[internal]  (7) [right       of=6]  {};
  \node[internal]  (8) [below       of=6]  {};
  \node[internal]  (9) [below       of=8]  {};
  \node[internal] (10) [right       of=9]  {};
  \node[internal] (11) [above right of=6]  {};
  \node[internal] (12) [below left  of=8]  {};
  \node[internal] (13) [left        of=8]  {};
  \node[internal] (14) [below left  of=13] {};
  \node[external] (15) [left        of=14] {};
  \node[internal] (16) [below left  of=5]  {};
  \path let
         \p1 = (15),
         \p2 = (0)
        in
         node[external] (17) at (\x1, \y2) {};
  \path let
         \p1 = (15),
         \p2 = (2)
        in
         node[external] (18) at (\x1, \y2) {};
  \node[external] (19) [right of=7]  {};
  \node[external] (20) [right of=10] {};
  \path (1) edge                             (5)
            edge[bend left]                 (11)
            edge[bend right]                (13)
            edge node[near start] (e1) {}   (17)
            edge node[near start] (e2) {}   (18)
        (3) edge                             (4)
        (4) edge[out=-45,in=45]              (8)
        (5) edge[bend right, looseness=0.5] (13)
            edge[bend right, looseness=0.5]  (6)
        (6) edge[bend left]                  (8)
            edge[bend left]                  (7)
            edge[bend left]                 (14)
        (7) edge node[near start] (e3) {}   (19)
       (10) edge[bend right, looseness=0.5] (12)
            edge[bend left,  looseness=0.5] (11)
            edge node[near start] (e4) {}   (20)
       (12) edge[bend left]                 (16)
       (14) edge node[near start] (e5) {}   (15)
            edge[bend right]                (12)
       (16) edge[bend left,  looseness=0.5]  (9)
            edge[bend right, looseness=0.5]  (3);
  \begin{pgfonlayer}{background}
   \node[draw=\borderColor,thick,rounded corners,fit = (3) (4) (9) (e1) (e2) (e3) (e4) (e5)] {};
  \end{pgfonlayer}
 \end{tikzpicture}
 \caption{An $\mathcal{F}$-gate with 5 dangling edges.}
 \label{fig:Fgate}
\end{figure}

Formally, such a notion is defined by an $\mathcal{F}$-gate~\cite{CLX09a, CLX10}.
An $\mathcal{F}$-gate is similar to a signature grid $(H, \mathcal{F}, \pi)$ except that $H = (V,E,D)$ is a graph with some dangling edges $D$.
The dangling edges define external variables for the $\mathcal{F}$-gate.
(See Figure~\ref{fig:Fgate} for an example.)
We denote the regular edges in $E$ by $1, 2, \dotsc, m$, and denote the dangling edges in $D$ by $m+1, \dotsc, m+n$.
Then we can define a function $\Gamma$ for this $\mathcal{F}$-gate as
\[
 \Gamma(y_1, y_2, \dotsc, y_n) = \sum_{x_1, x_2, \dotsc, x_m \in \{0, 1\}} H(x_1, x_2, \dotsc, x_m, y_1, \dotsc, y_n),
\]
where $(y_1, y_2, \dotsc, y_n) \in \{0, 1\}^n$ denotes an assignment on the dangling edges
and $H(x_1, x_2, \dotsc, x_m,$ $y_1, y_2, \dotsc, y_n)$ denotes the value of the signature grid on an assignment of all edges,
which is the product of evaluations at all internal vertices.
We also call this function the signature $\Gamma$ of the $\mathcal{F}$-gate.
An $\mathcal{F}$-gate with underlying graph $H$ is planar if the graph $H'$,
formed by introducing a new vertex $v$ and attaching each dangling edge to $v$, is also planar.
A planar $\mathcal{F}$-gate can be used in a planar signature grid as if it is just a single vertex with the particular signature.

Using the idea of planar $\mathcal{F}$-gates, we can reduce one planar Holant problem to another.
Suppose $g$ is the signature of some planar $\mathcal{F}$-gate.
Then $\PlHolant(\mathcal{F} \cup \{g\}) \leq_T \PlHolant(\mathcal{F})$.
The reduction is quite simple.
Given an instance of $\PlHolant(\mathcal{F} \cup \{g\})$,
by replacing every appearance of $g$ by the $\mathcal{F}$-gate,
we get an instance of $\PlHolant(\mathcal{F})$.
Since the signature of the $\mathcal{F}$-gate is $g$,
the Holant values for these two signature grids are identical.

We note that even for a very simple signature set $\mathcal{F}$,
the signatures for all planar $\mathcal{F}$-gates can be quite complicated and expressive.

\subsection{\texorpdfstring{\#}{Count-}CSP and Its Tractable Signatures} \label{subsec:CSP-Tractable}

An instance of $\CSP(\mathcal{F})$ has the following bipartite view.
Create a node for each variable and each constraint.
Connect a variable node to a constraint node if the variable appears in the constraint function.
This bipartite graph is also known as the \emph{incidence graph}~\cite{CKS01} or \emph{constraint graph}.
Under this view, we can see that
\[
 \CSP(\mathcal{F}) \equiv_T \holant{\mathcal{F}}{\EQ} \equiv_T \Holant(\mathcal{F} \cup \EQ),
\]
where $\EQ = \{=_1, =_2, =_3, \dotsc\}$ is the set of equality signatures of all arities.
This equivalence also holds for the planar versions of these frameworks.

For the \#CSP framework, the following two signature sets are tractable~\cite{CLX09a}.

\begin{definition}
 A $k$-ary function $f(x_1, \dotsc, x_k)$ is \emph{affine} if it has the form
 \[\lambda \chi_{Ax = 0} \cdot \sqrt{-1}^{\sum_{j=1}^n \langle \alpha_j, x \rangle},\]
 where $\lambda \in \mathbb{C}$, $x = \transpose{(x_1, x_2, \dotsc, x_k, 1)}$, $A$ is a matrix over $\mathbb{F}_2$, $\alpha_j$ is a vector over $\mathbb{F}_2$,
 and $\chi$ is a 0-1 indicator function such that $\chi_{Ax = 0}$ is 1 iff $A x = 0$.
 Note that the dot product $\langle \alpha_j, x \rangle$ is calculated over $\mathbb{F}_2$,
 while the summation $\sum_{j=1}^n$ on the exponent of $i = \sqrt{-1}$ is evaluated as a sum mod 4 of 0-1 terms.
 We use $\mathscr{A}$ to denote the set of all affine functions.
\end{definition}

Notice that there is no restriction on the number of rows in the matrix $A$.
The trivial case is when $A$ is the zero matrix so that $\chi_{A x = 0} = 1$ holds for all $x$.

\begin{definition}
 A function is of \emph{product type} if it can be expressed as a product of unary functions,
 binary equality functions $([1,0,1])$, and binary disequality functions $([0,1,0])$.
 We use $\mathscr{P}$ to denote the set of product type functions.
\end{definition}

An alternate definition for $\mathscr{P}$, implicit in~\cite{CLX11a},
is the tensor closure of signatures with support on two entries of complement indices.

It is easy to see (cf.~Lemma~A.1 in~\cite{HL13}, the full version of~\cite{HL12}) that if $f$ is a symmetric signature in $\mathscr{P}$,
then $f$ is either degenerate, binary disequality, or generalized equality (i.e.~$[a,0,\dotsc,0,b]$ for $a, b \in \mathbb{C}$).
It is known that the set of non-degenerate symmetric signatures in $\mathscr{A}$ is precisely the nonzero signatures
($\lambda \neq 0$) in $\mathscr{F}_1 \union \mathscr{F}_2 \union \mathscr{F}_3$ with arity at least two,
where $\mathscr{F}_1$, $\mathscr{F}_2$, and $\mathscr{F}_3$ are three families of signatures defined as
\begin{align*}
 \mathscr{F}_1 &= \left\{\lambda \left([1,0]^{\otimes k} + i^r [0, 1]^{\otimes k}\right) \st \lambda \in \mathbb{C}, k = 1, 2, \dotsc, r = 0, 1, 2, 3\right\},\\
 \mathscr{F}_2 &= \left\{\lambda \left([1,1]^{\otimes k} + i^r [1,-1]^{\otimes k}\right) \st \lambda \in \mathbb{C}, k = 1, 2, \dotsc, r = 0, 1, 2, 3\right\}, \text{ and}\\
 \mathscr{F}_3 &= \left\{\lambda \left([1,i]^{\otimes k} + i^r [1,-i]^{\otimes k}\right) \st \lambda \in \mathbb{C}, k = 1, 2, \dotsc, r = 0, 1, 2, 3\right\}.
\end{align*}

\noindent
We explicitly list all the signatures in $\mathscr{F}_1 \cup \mathscr{F}_2 \cup \mathscr{F}_3$ up to an arbitrary constant multiple from $\mathbb{C}$:

%%% Requires the package "enumitem", which is used to reduce the amount of vertical white space in the list
\parbox{0.61\textwidth}{
 \begin{enumerate}
  \item $[1, 0, \dotsc, 0, \pm 1]$; \hfill $(\mathscr{F}_1, r=0,2)$
  \item $[1, 0, \dotsc, 0, \pm i]$; \hfill $(\mathscr{F}_1, r=1,3)$
  \item $[1,  0, 1,  0, \dotsc,   0  \text{ or } 1]$; \hfill $(\mathscr{F}_2, r=0)$
  \item $[1, -i, 1, -i, \dotsc, (-i) \text{ or } 1]$; \hfill $(\mathscr{F}_2, r=1)$
  \item $[0,  1, 0,  1, \dotsc,   0  \text{ or } 1]$; \hfill $(\mathscr{F}_2, r=2)$
  \item $[1,  i, 1,  i, \dotsc,   i  \text{ or } 1]$; \hfill $(\mathscr{F}_2, r=3)$
  \item $[1,  0, -1,  0, 1,  0, -1,  0, \dotsc, 0 \text{ or } 1 \text{ or } (-1)]$; \hfill $(\mathscr{F}_3, r=0)$
  \item $[1,  1, -1, -1, 1,  1, -1, -1, \dotsc,               1 \text{ or } (-1)]$; \hfill $(\mathscr{F}_3, r=1)$
  \item $[0,  1,  0, -1, 0,  1,  0, -1, \dotsc, 0 \text{ or } 1 \text{ or } (-1)]$; \hfill $(\mathscr{F}_3, r=2)$
  \item $[1, -1, -1,  1, 1, -1, -1,  1, \dotsc,               1 \text{ or } (-1)]$. \hfill $(\mathscr{F}_3, r=3)$
 \end{enumerate}}

In the Holant framework, there are two corresponding signature sets that are tractable.
A signature $f$ is $\mathscr{A}$-transformable
if there exists a holographic transformation $T$ such that $f \in T \mathscr{A}$ and $[1,0,1] T^{\otimes 2} \in \mathscr{A}$.
Similarly, a signature $f$ is $\mathscr{P}$-transformable
if there exists a holographic transformation $T$ such that $f \in T \mathscr{P}$ and $[1,0,1] T^{\otimes 2} \in \mathscr{P}$.
These two families are tractable because after a transformation by $T$, it is a tractable \#CSP instance.
We note that $\widehat{\mathscr{A}} = \mathscr{A}$.
For symmetric signatures, this easily follows from the expressions of the signatures in $\mathscr{F}_1 \cup \mathscr{F}_2 \cup \mathscr{F}_3$.

\subsection{Matchgate Signatures}

Matchgates were introduced by Valiant~\cite{Val02a, Val02b} in order to give polynomial-time algorithms for a collection of counting problems over planar graphs.
As the name suggests,
problems expressible by matchgates can be reduced to computing a weighted sum of perfect matchings.
The latter problem is tractable over planar graphs by Kasteleyn's algorithm~\cite{Kas67}.
These counting problems are naturally expressed in the Holant framework using \emph{matchgate signatures}.
We use $\mathscr{M}$ to denote the set of all matchgate signatures; thus $\PlHolant(\mathscr{M})$ is tractable.
In general, matchgate signatures are characterized by the matchgate identities (see~\cite{CG13} for the identities and a self-contained proof).

The parity of a matchgate signature is even (resp.~odd) if its support is on entries of even (resp.~odd) Hamming weight.
Lemmas~6.2 and~6.3 in~\cite{CC07b} (and the paragraph the follows them) characterize the symmetric signatures in $\mathscr{M}$.
Instead of formally stating these two lemmas,
we explicitly list all the symmetric signatures in $\mathscr{M}$:
For any $\alpha, \beta \in \mathbb{C}$,
\begin{enumerate}
 \item $[\alpha^n, 0, \alpha^{n-1} \beta, 0, \dotsc, 0, \alpha \beta^{n-1}, 0, \beta^n]$;
 \item $[\alpha^n, 0, \alpha^{n-1} \beta, 0, \dotsc, 0, \alpha \beta^{n-1}, 0, \beta^n, 0]$;
 \item $[0, \alpha^n, 0, \alpha^{n-1} \beta, 0, \dotsc, 0, \alpha \beta^{n-1}, 0, \beta^n]$;
 \item $[0, \alpha^n, 0, \alpha^{n-1} \beta, 0, \dotsc, 0, \alpha \beta^{n-1}, 0, \beta^n, 0]$.
\end{enumerate}
Roughly speaking, the symmetric matchgate signatures have~0 for every other entry (which is called the \emph{parity condition}),
and form a geometric progression with the remaining entries.

In the standard basis of the \#CSP framework,
the set of signatures $\left[\begin{smallmatrix} 1 & 1 \\ 1 & -1 \end{smallmatrix}\right] \mathscr{M} = \widehat{\mathscr{M}}$ is tractable
and consists of signatures with the following expressions.\footnote{Even though Theorem~\ref{thm:matchgate_char:b2} is technically about generator signatures,
neither generators nor recognizers are mentioned because Theorems~3 and~4 in~\cite{CL10a} coincide when the basis is an orthogonal transformation.}

\begin{theorem}[Special case of Theorem~4 in~\cite{CL10a}] \label{thm:matchgate_char:b2}
 A symmetric signature $[f_0, f_1, \dotsc, f_n]$ is realizable under the basis $\left[\begin{smallmatrix} 1 & 1 \\ 1 & -1 \end{smallmatrix}\right]$
 iff it takes one of the following forms:
 \begin{enumerate}
  \item there exists constants $\lambda, \alpha, \beta \in \mathbb{C}$ and $\varepsilon = \pm 1$, such that for all $\ell$, $0 \le \ell \le n$,
  \label{case:hadamard_matchgate:general}
  \[f_\ell = \lambda [(\alpha + \beta)^{n-\ell} (\alpha - \beta)^\ell + \varepsilon (\alpha - \beta)^{n-\ell} (\alpha + \beta)^\ell];\]
  \item there exists a constant $\lambda \in \mathbb{C}$, such that for all $\ell$, $0 \le \ell \le n$,
  \label{case:hadamard_matchgate:pm_reverse}
  \[f_\ell = \lambda (n - 2 \ell) (-1)^\ell;\]
  \item there exists a constant $\lambda \in \mathbb{C}$, such that for all $\ell$, $0 \le \ell \le n$,
  \label{case:hadamard_matchgate:pm}
  \[f_\ell = \lambda (n - 2 \ell).\]
 \end{enumerate}
\end{theorem}

We note that case~\ref{case:hadamard_matchgate:general} corresponds to the general case
($\varepsilon = +1$ for signatures with even parity and $\varepsilon = -1$ for signatures with odd parity)
while case~\ref{case:hadamard_matchgate:pm} corresponds to the perfect matching signatures $[0,1,0,\dotsc,0]$
and case~\ref{case:hadamard_matchgate:pm_reverse} corresponds to their reversals.

We summarize the known tractability results for the $\PlCSP$ framework in the following theorem,
which is stated in the Hadamard basis with $[1,0]$ and $[0,1]$ present.

\begin{theorem} \label{thm:PlCSP_tractable:Hadamard}
 Let $\mathcal{F}$ be any set of symmetric, complex-valued signatures in Boolean variables.
 Then $\PlHolant^c(\mathcal{F} \union \widehat{\EQ})$ is tractable if
 $\mathcal{F} \subseteq \mathscr{A}$, $\mathcal{F} \subseteq \widehat{\mathscr{P}}$, or $\mathcal{F} \subseteq \mathscr{M}$.
\end{theorem}

We also say a signature $f$ is $\mathscr{M}$-transformable
if there exists a holographic transformation $T$ such that $f \in T \mathscr{M}$ and $[1,0,1] T^{\otimes 2} \in \mathscr{M}$.

\subsection{Some Known Dichotomies}

We use the dichotomy for a single ternary signature in the Holant framework to prove the dichotomy for a single arity~4 signature.
A signature is called \emph{vanishing} if the Holant of any signature grid using only that signature is zero (see~\cite{CGW13}, the full version of~\cite{CGW12}).

\begin{theorem}[Special case of Theorem~V.1 in~\cite{CLX10}] \label{thm:arity3:singleton}
 If $f$ is a symmetric, non-degenerate, complex-valued ternary signature,
 then $\PlHolant(f)$ is $\SHARPP$-hard unless $f$ satisfies one of the following conditions,
 in which case the problem is in $\P$:
 \begin{enumerate}
  \item $\Holant(f)$ is tractable (i.e. $f$ is $\mathscr{A}$-transformable, $\mathscr{P}$-transformable, or vanishing);
  \item $f$ is $\mathscr{M}$-transformable.
 \end{enumerate}
\end{theorem}

We use the following theorem about edge-weighted signatures on degree prescribed graphs in both of our dichotomy theorems.
See also Theorem~22 in~\cite{Kow10}, which contains a proof.

\begin{theorem}[Theorem~4 in~\cite{CK11}] \label{thm:degree_prescribed_homomorphism}
 Let $S \subseteq \mathbb{Z}^+$ be nonempty, let $\mathcal{G} = \{=_k \st k \in S\}$, and let $d = \gcd(S)$.
 Then $\plholant{[f_0, f_1, f_2]}{\mathcal{G}}$ is $\SHARPP$-hard for all $f_0, f_1, f_2 \in \mathbb{C}$
 unless one of the following conditions hold, in which case the problem is in $\P$:
 \begin{enumerate}
  \item $\mathcal{G} \subseteq \{=_1, =_2\}$;
  \item $f_0 f_2 = f_1^2$;
  \item $f_1 = 0$;
  \item $f_0 f_2 = -f_1^2 \wedge f_0^d = -f_2^d$;
  \item $f_0^d = f_2^d$.
 \end{enumerate}
\end{theorem}

For the arity~4 dichotomy,
we use Theorem~\ref{thm:degree_prescribed_homomorphism} with $\mathcal{G} = \{=_4\}$.
For the $\PlCSP$ dichotomy,
we use Theorem~\ref{thm:degree_prescribed_homomorphism} with $\mathcal{G} = \EQ$,
which is the special case of $\PlCSP(\mathcal{F})$ when $\mathcal{F}$ contains a single binary signature.
Over general domains,
this special case is also known as counting graph homomorphism from a planar input graph to a fixed target graph.
Furthermore, we perform a holographic transformation by the Hadamard matrix $H = \left[\begin{smallmatrix} 1 & 1 \\ 1 & -1 \end{smallmatrix}\right]$.
Under this transformation, it is easy to see that the conditions $f_0 f_2 = f_1^2$ and $f_0 f_2 = -f_1^2 \wedge f_0 = -f_2$ are invariant
while the conditions $f_1 = 0$ and $f_0 = f_2$ map to each other.
Therefore, by an apparent coincidence, the tractability conditions remain the same.
To be clear, we restate Theorem~\ref{thm:degree_prescribed_homomorphism} both before and after a holographic transformation by $H$ with $\mathcal{G} = \EQ$.

\begin{theorem}[Special case of Theorem~\ref{thm:degree_prescribed_homomorphism}] \label{thm:Pl-Graph_Homomorphism}
 For any $f_0, f_1, f_2 \in \mathbb{C}$,
 both $\plholant{[f_0, f_1, f_2]}{\EQ}$ and $\plholant{[f_0, f_1, f_2]}{\widehat{\EQ}}$ are $\SHARPP$-hard unless one of the following conditions hold,
 in which case both problems are in $\P$:
 \begin{enumerate}
  \item $f_0 f_2 = f_1^2$;                   \label{case:Pl-GH:degenerate}
  \item $f_1 = 0$;                           \label{case:Pl-GH:M}
  \item $f_0 f_2 = -f_1^2$ and $f_0 = -f_2$; \label{case:Pl-GH:A}
  \item $f_0 = f_2$.                         \label{case:Pl-GH:P-hat}
 \end{enumerate}
\end{theorem}

\section{The Complexity of Counting Eulerian Orientations} \label{sec:EO}

Recall the definition of an Eulerian orientation.

\begin{definition}
 Given a graph $G$, an orientation of its edges is an \emph{Eulerian orientation} if for each vertex $v$ of $G$,
 the number of incoming edges of $v$ equals the number of outgoing edges of $v$.
\end{definition}

Counting the number of (unweighted) Eulerian orientations over 4-regular graphs was shown to be $\SHARPP$-hard in Theorem~V.10 of~\cite{HL12}.
We improve this result by showing that this problem remains $\SHARPP$-hard when the input is also planar.
The reduction begins with the problem of evaluating the Tutte polynomial at the point~(3,3), which is $\SHARPP$-hard even for planar graphs.

\begin{theorem}[Theorem~5.1 in~\cite{Ver05}] \label{thm:Tutte}
 For any $x, y \in \mathbb{C}$,
 the problem of computing the Tutte polynomial at $(x,y)$ over planar graphs is $\SHARPP$-hard
 unless $(x - 1) (y - 1) \in \{1, 2\}$ or $(x,y) \in \{(1,1), (-1, -1), (j,j^2), (j^2, j)\}$,
 where $j = e^{2 \pi i / 3}$.
 In each of these exceptional cases,
 the computation can be done in polynomial time.
\end{theorem}

The first step in the reduction concerns a sum of weighted Eulerian orientations on a medial graph of a planar graph.
Recall the definition of a medial graph.

\begin{definition}[cf.~\cite{BO92}]
 For a connected plane graph $G$ (i.e.~a planar embedding of a planar graph),
 its \emph{medial graph} $H$ has a vertex for each edge of $G$
 and two vertices in $H$ are joined by an edge for each face of $G$ in which their corresponding edges occur consecutively.
\end{definition}

\begin{figure}[t]
 \centering
 \def\medialNodeDist{2.5cm}
 \tikzstyle{open}   = [draw, black, fill=white, shape=circle]
 \tikzstyle{closed} = [draw,        fill,       shape=circle]
 \subcaptionbox{\label{subfig:planar_graph}}{
  \centering
  \begin{tikzpicture}[scale=\scale,transform shape,node distance=\medialNodeDist,semithick]
   \node[closed] (0)              {};
   \node[closed] (1) [right of=0] {};
   \node[closed] (2) [above of=0] {};
   \node[closed] (3) [above of=1] {};
   \node[closed] (4) [above of=2] {};
   \path (0) edge[out=-45, in=-135]               node[external] (m0) {} (1)
             edge[out= 45, in= 135]               node[external] (m1) {} (1)
             edge                                 node[external] (m2) {} (2)
         (1) edge                                 node[external] (m3) {} (3)
         (2) edge                                 node[external] (m4) {} (3)
             edge                                 node[external] (m5) {} (4)
         (3) edge[out=125, in=  55, looseness=30] node[external] (m6) {} (3);
   \path (m0) edge[white, densely dashed, out= 135, in=-135]                (m1)
              edge[white, densely dashed, out=  45, in= -45]                (m1)
              edge[white, densely dashed, out=-145, in=-135, looseness=1.7] (m2)
              edge[white, densely dashed, out= -35, in= -45, looseness=1.7] (m3)
         (m1) edge[white, densely dashed]                                   (m2)
              edge[white, densely dashed]                                   (m3)
         (m2) edge[white, densely dashed]                                   (m4)
              edge[white, densely dashed, out= 135, in=-135]                (m5)
         (m3) edge[white, densely dashed]                                   (m4)
              edge[white, densely dashed, out=  45, in=  15]                (m6)
         (m4) edge[white, densely dashed]                                   (m5)
              edge[white, densely dashed, out=  90, in= 165]                (m6)
         (m5) edge[white, densely dashed, out= 125, in=  55, looseness=30]  (m5)
         (m6) edge[white, densely dashed, out=-125, in= -55, looseness=15]  (m6);
  \end{tikzpicture}}
 \qquad
 \qquad
 \subcaptionbox{\label{subfig:superimposed}}{
  \centering
  \begin{tikzpicture}[scale=\scale,transform shape,node distance=\medialNodeDist,semithick]
   \node[closed] (0)              {};
   \node[closed] (1) [right of=0] {};
   \node[closed] (2) [above of=0] {};
   \node[closed] (3) [above of=1] {};
   \node[closed] (4) [above of=2] {};
   \path (0) edge[out=-45, in=-135]               node[open] (m0) {} (1)
             edge[out= 45, in= 135]               node[open] (m1) {} (1)
             edge                                 node[open] (m2) {} (2)
         (1) edge                                 node[open] (m3) {} (3)
         (2) edge                                 node[open] (m4) {} (3)
             edge                                 node[open] (m5) {} (4)
         (3) edge[out=125, in=  55, looseness=30] node[open] (m6) {} (3);
   \path (m0) edge[densely dashed, out= 135, in=-135]                (m1)
              edge[densely dashed, out=  45, in= -45]                (m1)
              edge[densely dashed, out=-145, in=-135, looseness=1.7] (m2)
              edge[densely dashed, out= -35, in= -45, looseness=1.7] (m3)
         (m1) edge[densely dashed]                                   (m2)
              edge[densely dashed]                                   (m3)
         (m2) edge[densely dashed]                                   (m4)
              edge[densely dashed, out= 135, in=-135]                (m5)
         (m3) edge[densely dashed]                                   (m4)
              edge[densely dashed, out=  45, in=  15]                (m6)
         (m4) edge[densely dashed]                                   (m5)
              edge[densely dashed, out=  90, in= 165]                (m6)
         (m5) edge[densely dashed, out= 125, in=  55, looseness=30]  (m5)
         (m6) edge[densely dashed, out=-125, in= -55, looseness=15]  (m6);
  \end{tikzpicture}}
 \qquad
 \qquad
 \subcaptionbox{\label{subfig:medial_graph}}{
  \centering
  \begin{tikzpicture}[scale=\scale,transform shape,node distance=\medialNodeDist,semithick]
   \node[external] (0)              {};
   \node[external] (1) [right of=0] {};
   \node[external] (2) [above of=0] {};
   \node[external] (3) [above of=1] {};
   \node[external] (4) [above of=2] {};
   \path (0) edge[white, out=-45, in=-135]               node[open] (m0) {} (1)
             edge[white, out= 45, in= 135]               node[open] (m1) {} (1)
             edge[white]                                 node[open] (m2) {} (2)
         (1) edge[white]                                 node[open] (m3) {} (3)
         (2) edge[white]                                 node[open] (m4) {} (3)
             edge[white]                                 node[open] (m5) {} (4)
         (3) edge[white, out=125, in=  55, looseness=30] node[open] (m6) {} (3);
   \path (m0) edge[densely dashed, out= 135, in=-135]                (m1)
              edge[densely dashed, out=  45, in= -45]                (m1)
              edge[densely dashed, out=-145, in=-135, looseness=1.7] (m2)
              edge[densely dashed, out= -35, in= -45, looseness=1.7] (m3)
         (m1) edge[densely dashed]                                   (m2)
              edge[densely dashed]                                   (m3)
         (m2) edge[densely dashed]                                   (m4)
              edge[densely dashed, out= 135, in=-135]                (m5)
         (m3) edge[densely dashed]                                   (m4)
              edge[densely dashed, out=  45, in=  15]                (m6)
         (m4) edge[densely dashed]                                   (m5)
              edge[densely dashed, out=  90, in= 165]                (m6)
         (m5) edge[densely dashed, out= 125, in=  55, looseness=30]  (m5)
         (m6) edge[densely dashed, out=-125, in= -55, looseness=15]  (m6);
  \end{tikzpicture}}
 \caption{A plane graph~(\protect\subref{subfig:planar_graph}), its medial graph~(\protect\subref{subfig:medial_graph}),
 and the two graphs superimposed~(\protect\subref{subfig:superimposed}).}
 %the \protect command changes the fragile command \subref into a robust one.  See http://www.tex.ac.uk/cgi-bin/texfaq2html?label=extrabrace
 \label{fig:medial_graph_example}
\end{figure}
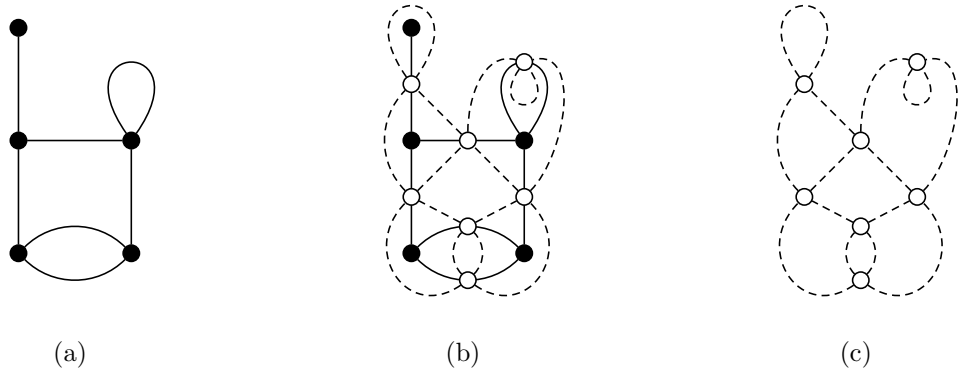

An example of a plane graph and its medial graph are given in Figure~\ref{fig:medial_graph_example}.
Notice that a medial graph of a planar graph is always a planar 4-regular graph.
Las Vergnas~\cite{Ver88} connected the evaluation of the Tutte polynomial of a graph $G$ at the point~(3,3)
with a sum of weighted Eulerian orientations on a medial graph of $G$.

\begin{theorem}[Theorem~2.1 in~\cite{Ver88}] \label{thm:TutteToWEO}
 Let $G$ be a connected plane graph and let $\mathscr{O}(H)$ be the set of all Eulerian orientations in the medial graph $H$ of $G$.
 Then
 \begin{align}
  2 \cdot \operatorname{T}(G; 3, 3) = \sum_{O \in \mathscr{O}(H)} 2^{\beta(O)}, \label{equ:TutteAndWEO}
 \end{align}
 where $\beta(O)$ is the number of saddle vertices in the orientation $O$,
 i.e.~the number of vertices in which the edges are oriented ``in, out, in, out'' in cyclic order.
\end{theorem}

Although the medial graph depends on a particular embedding of the planar graph $G$,
the right-hand side of~(\ref{equ:TutteAndWEO}) is invariant under different embeddings of $G$.
This follows from~(\ref{equ:TutteAndWEO}) and the fact that the Tutte polynomial does not depend on the embedding of $G$.

In addition to these two theorems,
our proof also uses two definitions from~\cite{CGW13}.

\begin{definition}[Definition~6.1 in~\cite{CGW13}]
 A 4-by-4 matrix is \emph{redundant} if its middle two rows and middle two columns are the same.
\end{definition}

An example of a redundant matrix is the signature matrix of a symmetric arity~4 signature.

\begin{definition}[Definition~6.2 in~\cite{CGW13}]
 The \emph{signature matrix} of a symmetric arity~4 signature $f = [f_0, f_1, f_2, f_3, f_4]$ is
 \begin{align*}
  M_f =
  \begin{bmatrix}
   f_0 & f_1 & f_1 & f_2\\
   f_1 & f_2 & f_2 & f_3\\
   f_1 & f_2 & f_2 & f_3\\
   f_2 & f_3 & f_3 & f_4
  \end{bmatrix}.
 \end{align*}
 This definition extends to an asymmetric signature $g$ as
 \begin{align*}
  M_g =
  \begin{bmatrix}
   g^{0000} & g^{0010} & g^{0001} & g^{0011}\\
   g^{0100} & g^{0110} & g^{0101} & g^{0111}\\
   g^{1000} & g^{1010} & g^{1001} & g^{1011}\\
   g^{1100} & g^{1110} & g^{1101} & g^{1111}
  \end{bmatrix}.
 \end{align*}
 When we present $g$ as an $\mathcal{F}$-gate, we order the four external edges ABCD counterclockwise.
 In $M_g$, the row index bits are ordered AB and the column index bits are ordered DC, in a reverse way.
 This is for convenience so that the signature matrix of the linking of two arity 4 $\mathcal{F}$-gates
 is the matrix product of the signature matrices of the two $\mathcal{F}$-gates.

 If $M_g$ is redundant, we also define the \emph{compressed signature matrix} of $g$ as
 \[
  \widetilde{M_g}
  =
  \begin{bmatrix}
   1 & 0 & 0 & 0\\
   0 & \frac{1}{2} & \frac{1}{2} & 0\\
   0 & 0 & 0 & 1
  \end{bmatrix}
  M_g
  \begin{bmatrix}
   1 & 0 & 0\\
   0 & 1 & 0\\
   0 & 1 & 0\\
   0 & 0 & 1
  \end{bmatrix}.
 \]
\end{definition}

Now we can prove our hardness result.

\begin{theorem} \label{thm:4reg_planar_EO_hard}
 \textsc{\#Eulerian-Orientations} is $\SHARPP$-hard for planar 4-regular graphs.
\end{theorem}

\begin{proof}
 We reduce calculating the right-hand side of~(\ref{equ:TutteAndWEO}) to $\plholant{{\ne}_2}{[0,0,1,0,0]}$,
 which denotes the problem of counting the number of Eulerian orientations over planar 4-regular graphs as a bipartite Holant problem.
 Then by Theorem~\ref{thm:Tutte} and Theorem~\ref{thm:TutteToWEO}, we conclude that $\plholant{{\ne}_2}{[0,0,1,0,0]}$ is $\SHARPP$-hard.
 
 The right-hand side of~(\ref{equ:TutteAndWEO}) is the bipartite Holant problem $\plholant{{\ne_2}}{f}$,
 where the signature matrix of $f$ is
 \[
  M_f
  =
  \begin{bmatrix}
   0 & 0 & 0 & 1 \\
   0 & 1 & 2 & 0 \\
   0 & 2 & 1 & 0 \\
   1 & 0 & 0 & 0
  \end{bmatrix}.
 \]
 We perform a holographic transformation by $Z = \left[\begin{smallmatrix} 1 & 1 \\ i & -i \end{smallmatrix}\right]$ to get
 \begin{align*}
  \plholant{{\ne}_2}{f}
  &\equiv_T \plholant{[0,1,0] (Z^{-1})^{\otimes 2}}{Z^{\otimes 4} f}\\
  &\equiv_T \plholant{[1,0,1] / 2}{4 \hat{f}}\\
  &\equiv_T \PlHolant(\hat{f}),
 \end{align*}
 where the signature matrix of $\hat{f}$ is
 \[
  M_{\hat{f}}
  =
  \begin{bmatrix}
   2 & 0 & 0 & 1 \\
   0 & 1 & 0 & 0 \\
   0 & 0 & 1 & 0 \\
   1 & 0 & 0 & 2
  \end{bmatrix}.
 \]
 We also perform the same holographic transformation by $Z$ on our target counting problem $\plholant{{\ne}_2}{[0,0,1,0,0]}$ to get
 \begin{align*}
  \plholant{{\ne}_2}{[0,0,1,0,0]}
  &\equiv_T \plholant{[0,1,0] (Z^{-1})^{\otimes 2}}{Z^{\otimes 4} [0,0,1,0,0]}\\
  &\equiv_T \plholant{[1,0,1] / 2}{2 [3,0,1,0,3]}\\
  &\equiv_T \PlHolant([3,0,1,0,3]).
 \end{align*}
 Using the planar tetrahedron gadget in Figure~\ref{fig:gadget:planar_tetrahedron},
 we assign $[3,0,1,0,3]$ to every vertex and obtain a gadget with signature $32 \hat{g}$,
 where the signature matrix of $\hat{g}$ is
 \[
  M_{\hat{g}}
  =
  \frac{1}{2}
  \begin{bmatrix}
   19 & 0 & 0 &  7 \\
    0 & 7 & 5 &  0 \\
    0 & 5 & 7 &  0 \\
    7 & 0 & 0 & 19
  \end{bmatrix}.
 \]

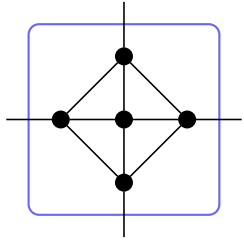
\begin{figure}[t]
 \centering
 \begin{tikzpicture}[scale=\scale,transform shape,node distance=\nodeDist,semithick]
  \node[external] (0)                    {};
  \node[internal] (1) [right       of=0] {};
  \node[internal] (2) [right       of=1] {};
  \node[internal] (3) [above       of=2] {};
  \node[external] (4) [above       of=3] {};
  \node[internal] (5) [below       of=2] {};
  \node[external] (6) [below       of=5] {};
  \node[internal] (7) [right       of=2] {};
  \node[external] (8) [right       of=7] {};
  \path (1) edge (2)
            edge (3)
            edge (5)
        (2) edge (3)
            edge (5)
            edge (7)
        (3) edge (7)
        (5) edge (7);
  \path (0) edge node[near end]   (e1) {} (1)
        (3) edge node[near start] (e2) {} (4)
        (5) edge node[near start] (e3) {} (6)
        (7) edge node[near start] (e4) {} (8);
  \begin{pgfonlayer}{background}
   \node[draw=\borderColor,thick,rounded corners,fit = (e1) (e2) (e3) (e4)] {};
  \end{pgfonlayer}
 \end{tikzpicture}
 \caption{The planar tetrahedron gadget. Each vertex is assigned $[3, 0, 1, 0, 3]$.}
 \label{fig:gadget:planar_tetrahedron}
\end{figure}

\begin{figure}[b]
 \centering
 \captionsetup[subfigure]{labelformat=empty}
 \subcaptionbox{$N_1$}{
  \begin{tikzpicture}[scale=\scale,transform shape,node distance=\nodeDist,semithick]
   \node[external] (0)                    {};
   \node[external] (1) [right       of=0] {};
   \node[internal] (2) [below right of=1] {};
   \node[external] (3) [below left  of=2] {};
   \node[external] (4) [left        of=3] {};
   \node[external] (5) [above right of=2] {};
   \node[external] (6) [right       of=5] {};
   \node[external] (7) [below right of=2] {};
   \node[external] (8) [right       of=7] {};
   \path (0) edge[out=   0, in=135] (2)
         (2) edge[out=-135, in=  0] (4)
             edge[out=  45, in=180] (6)
             edge[out= -45, in=180] (8);
   \begin{pgfonlayer}{background}
    \node[draw=\borderColor,thick,rounded corners,fit = (1) (3) (5) (7)] {};
   \end{pgfonlayer}
  \end{tikzpicture}}
 \qquad
 \subcaptionbox{$N_2$}{
  \begin{tikzpicture}[scale=\scale,transform shape,node distance=\nodeDist,semithick]
   \node[external]  (0)                    {};
   \node[external]  (1) [right       of=0] {};
   \node[internal]  (2) [below right of=1] {};
   \node[external]  (3) [below left  of=2] {};
   \node[external]  (4) [left        of=3] {};
   \node[external]  (5) [right       of=2] {};
   \node[internal]  (6) [right       of=5] {};
   \node[external]  (7) [above right of=6] {};
   \node[external]  (8) [right       of=7] {};
   \node[external]  (9) [below right of=6] {};
   \node[external] (10) [right       of=9] {};
   \path (0) edge[out=   0, in= 135]  (2)
         (2) edge[out=-135, in=   0]  (4)
             edge[out=  45, in= 135]  (6)
             edge[out= -45, in=-135]  (6)
         (6) edge[out=  45, in= 180]  (8)
             edge[out= -45, in= 180] (10);
   \begin{pgfonlayer}{background}
    \node[draw=\borderColor,thick,rounded corners,fit = (1) (3) (7) (9)] {};
   \end{pgfonlayer}
  \end{tikzpicture}}
 \qquad
 \subcaptionbox{$N_{s+1}$}{
  \begin{tikzpicture}[scale=\scale,transform shape,node distance=\nodeDist,semithick]
   \node[external]  (0)                     {};
   \node[external]  (1) [above left  of=0]  {};
   \node[external]  (2) [below left  of=0]  {};
   \node[external]  (3) [below left  of=1]  {};
   \node[external]  (4) [below left  of=3]  {};
   \node[external]  (5) [above left  of=3]  {};
   \node[external]  (6) [left        of=4]  {};
   \node[external]  (7) [left        of=5]  {};
   \node[external]  (8) [right       of=0]  {};
   \node[internal]  (9) [right       of=8]  {};
   \node[external] (10) [above right of=9]  {};
   \node[external] (11) [below right of=9]  {};
   \node[external] (12) [right       of=10] {};
   \node[external] (13) [right       of=11] {};
   \path let
          \p1 = (1),
          \p2 = (2)
         in
          node[external] at (\x1, \y1 / 2 + \y2 / 2) {\Huge $N_s$};
   \path let
          \p1 = (0)
         in
          node[external] (14) at (\x1 + 2, \y1 + 10) {};
   \path let
          \p1 = (0)
         in
          node[external] (15) at (\x1 + 2, \y1 - 10) {};
   \path let
          \p1 = (3)
         in
          node[external] (16) at (\x1 - 2, \y1 + 10) {};
   \path let
          \p1 = (3)
         in
          node[external] (17) at (\x1 - 2, \y1 - 10) {};
   \path (7) edge[out=   0, in=135] (16)
        (17) edge[out=-135, in=  0]  (6)
        (14) edge[out=  35, in=135]  (9)
         (9) edge[out=-135, in=-35] (15)
             edge[out=  45, in=180] (12)
             edge[out= -45, in=180] (13);
   \begin{pgfonlayer}{background}
    \node[draw=\borderColor,thick,densely dashed,rounded corners,fit = (0) (1.south) (2.north) (3)] {};
    \node[draw=\borderColor,thick,rounded corners,fit = (4) (5) (10) (11)] {};
   \end{pgfonlayer}
  \end{tikzpicture}}
 \caption{Recursive construction to interpolate $\hat{f}$. The vertices are assigned $\hat{g}$.}
 \label{fig:gadget:arity4:interpolate_WEO}
\end{figure}
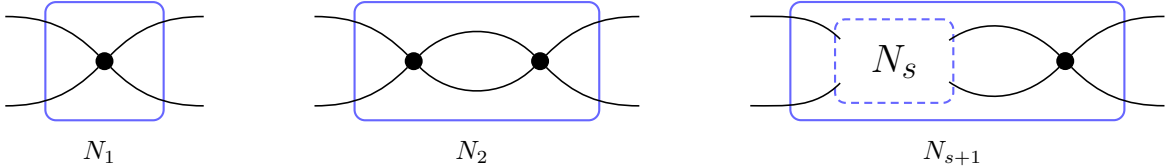
 
 Now we show how to reduce $\PlHolant(\hat{f})$ to $\PlHolant(\hat{g})$ by interpolation. 
 Consider an instance $\Omega$ of $\PlHolant(\hat{f})$.
 Suppose that $\hat{f}$ appears $n$ times in $\Omega$.
 We construct from $\Omega$ a sequence of instances $\Omega_s$ of $\Holant(\hat{g})$ indexed by $s \ge 1$.
 We obtain $\Omega_s$ from $\Omega$ by replacing each occurrence of $\hat{f}$ with the gadget $N_s$ in
 Figure~\ref{fig:gadget:arity4:interpolate_WEO} with $\hat{g}$ assigned to all vertices.
 Although $\hat{f}$ and $\hat{g}$ are asymmetric signatures, they are invariant under a cyclic permutation of their inputs.
 Thus, it is unnecessary to specify which edge corresponds to which input.
 We call such signatures \emph{rotationally symmetric}.
 
 To obtain $\Omega_s$ from $\Omega$,
 we effectively replace $M_{\hat{f}}$ with $M_{N_s} = (M_{\hat{g}})^s$,
 the $s$th power of the signature matrix $M_{\hat{g}}$.
 Let
 \[
  T
  =
  \begin{bmatrix}
    0 & 0 &  1 & 1 \\
    1 & 1 &  0 & 0 \\
   -1 & 1 &  0 & 0 \\
    0 & 0 & -1 & 1
  \end{bmatrix}.
 \]
 Then
 \[
  M_{\hat{f}}
  =
  T \Lambda_{\hat{f}} T^{-1}
  =
  T
  \begin{bmatrix}
   1 & 0 & 0 & 0 \\
   0 & 1 & 0 & 0 \\
   0 & 0 & 1 & 0 \\
   0 & 0 & 0 & 3
  \end{bmatrix}
  T^{-1}
  \qquad \text{and} \qquad
  M_{\hat{g}}
  =
  T \Lambda_{\hat{g}} T^{-1}
  =
  T
  \begin{bmatrix}
   1 & 0 & 0 &  0 \\
   0 & 6 & 0 &  0 \\
   0 & 0 & 6 &  0 \\
   0 & 0 & 0 & 13
  \end{bmatrix}
  T^{-1}.
 \]
 We can view our construction of $\Omega_s$ as first replacing each $M_{\hat{f}}$ by $T \Lambda_{\hat{f}} T^{-1}$ to obtain a signature grid $\Omega'$,
 which does not change the Holant value,
 and then replacing each $\Lambda_{\hat{f}}$ with $\Lambda_{\hat{g}}^s$.
 We stratify the assignments in $\Omega'$ based on the assignment to $\Lambda_{\hat{f}}$.
 We only need to consider the assignments to $\Lambda_{\hat{f}}$ that assign
 \begin{itemize}
  \item 0000 $j$ many times,
  \item 0110 or 1001 $k$ many times, and
  \item 1111 $\ell$ many times.
 \end{itemize}
 Let $c_{j k \ell}$ be the sum over all such assignments of the products of evaluations from $T$ and $T^{-1}$ but excluding $\Lambda_{\hat{f}}$ on $\Omega'$.
 Then
 \[\PlHolant_\Omega = \sum_{j + k + \ell = n} 3^\ell c_{j k \ell}\]
 and the value of the Holant on $\Omega_s$, for $s \ge 1$, is
 \[\PlHolant_{\Omega_s} = \sum_{j + k + \ell = n} (6^k 13^\ell)^s c_{j k \ell}.\]
 This coefficient matrix in the linear system involving $\PlHolant_{\Omega_s}$ is Vandermonde and of full rank
 since for any $0 \le k + \ell \le n$ and $0 \le k' + \ell' \le n$ such that $(k, \ell) \ne (k', \ell')$, $6^k 13^\ell \neq 6^{k'} 13^{\ell'}$.
 Therefore, we can solve the linear system for the unknown $c_{j k \ell}$'s and obtain the value of $\Holant_\Omega$.
\end{proof}

The previous proof can be easily modified to reduce from \#EO over 4-regular graphs by interpolating the so-called crossover signature.
Conceptually, the current proof is simpler because
the $\SHARPP$-hardness proof for \#EO over 4-regular graphs in~\cite{HL12} reduces from the same starting point as our current proof.

One of our main results in this paper is a dichotomy for $\PlHolant(f)$ when $f$ is a symmetric arity~4 signature with complex weights.
This dichotomy uses the $\SHARPP$-hardness of counting Eulerian orientations over planar 4-regular graphs in a crucial way.
In~\cite{CGW13},
it was shown that most arity~4 signatures define a $\SHARPP$-hard Holant problem by a reduction from counting Eulerian orientations over 4-regular graphs
(see Lemmas~6.4 and~6.6 in~\cite{CGW13}).
Although the reductions were planar,
$\SHARPP$-hardness over planar 4-regular graphs did not follow because the complexity of counting Eulerian orientations over such graphs was unknown.
Theorem~\ref{thm:4reg_planar_EO_hard} shows that this problem is $\SHARPP$-hard.
Therefore, we obtain the planar version of Corollary~6.7 in~\cite{CGW13}.

\begin{corollary} \label{cor:arity4:nonsingular_compressed_hard}
 Let $f$ be an arity 4 signature with complex weights.
 If $M_f$ is redundant and $\widetilde{M_f}$ is nonsingular,
 then $\PlHolant(f)$ is $\SHARPP$-hard.
\end{corollary}

There is a simpler corollary for symmetric signatures.

\begin{corollary} \label{cor:arity4:nonsingular_compressed_hard:symmetric}
 For a symmetric arity~4 signature $[f_0, f_1, f_2, f_3, f_4]$ with complex weights,
 if there does \emph{not} exist $a, b, c \in \mathbb{C}$, not all zero, such that for all $k \in \{0,1,2\}$,
 \[a f_k + b f_{k+1} + c f_{k+2} = 0,\]
 then $\PlHolant(f)$ is $\SHARPP$-hard.
\end{corollary}

\begin{proof}
 If the compressed signature matrix $\widetilde{M_f}$ is nonsingular,
 then $\PlHolant(f)$ is $\SHARPP$-hard by Corollary~\ref{cor:arity4:nonsingular_compressed_hard},
 so assume that the rank of $\widetilde{M_f}$ is at most~2.
 Then we have
 \[
      a' \begin{pmatrix} f_0 \\ f_1 \\ f_2 \end{pmatrix}
  + 2 b' \begin{pmatrix} f_1 \\ f_2 \\ f_3 \end{pmatrix}
  +   c' \begin{pmatrix} f_2 \\ f_3 \\ f_4 \end{pmatrix}
  =      \begin{pmatrix} 0   \\ 0   \\ 0   \end{pmatrix}
 \]
 for some $a',b',c' \in \mathbb{C}$, not all zero.
 Thus, $a = a'$, $b = 2 b'$, and $c = c'$ have the desired property.
\end{proof}

\begin{figure}[t]
 \centering
 \begin{tikzpicture}[scale=\scale,transform shape,node distance=\nodeDist,semithick]
  \node[internal]  (0)                    {};
  \node[external]  (1) [above left  of=0] {};
  \node[external]  (2) [below left  of=0] {};
  \node[external]  (3) [left        of=1] {};
  \node[external]  (4) [left        of=2] {};
  \node[external]  (5) [right       of=0] {};
  \node[external]  (6) [right       of=5] {};
  \node[internal]  (7) [right       of=6] {};
  \node[external]  (8) [above right of=7] {};
  \node[external]  (9) [below right of=7] {};
  \node[external] (10) [right       of=8] {};
  \node[external] (11) [right       of=9] {};
  \path (0) edge[out= 135, in=   0]                                     (3)
            edge[out=-135, in=   0]                                     (4)
            edge[out=  45, in= 135]                                     (7)
            edge                                                        (7)
            edge[out= -45, in=-135]                                     (7)
        (7) edge[out=  45, in= 180]                                    (10)
            edge[out= -45, in= 180]                                    (11);
  \begin{pgfonlayer}{background}
   \node[draw=\borderColor,thick,rounded corners,fit = (1) (2) (8) (9)] {};
  \end{pgfonlayer}
 \end{tikzpicture}
 \caption{The circles are assigned $[a,0,0,0,b,c]$.}
 \label{fig:gadget:arity5_232}
\end{figure}
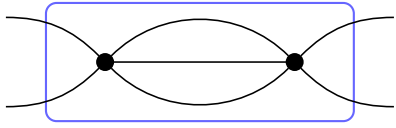

We close this section with a simple application of Corollary~\ref{cor:arity4:nonsingular_compressed_hard} to an arity~5.
We encounter signatures of this form in Sections~\ref{sec:pinning} and~\ref{sec:dichotomy}.

\begin{lemma} \label{lem:arity5:hard_sig}
 Let $a, b, c \in \mathbb{C}$.
 If $a b \ne 0$,
 then for any set $\mathcal{F}$ of complex-weighted symmetric signatures containing $[a,0,0,0,b,c]$,
 $\PlHolant(\mathcal{F})$ is $\SHARPP$-hard.
\end{lemma}

\begin{proof}
 Let $f$ be the signature of the gadget in Figure~\ref{fig:gadget:arity5_232} with $[a,0,0,0,b,c]$ assigned to both vertices. 
 The signature matrix of $f$ is
 \[
  \begin{bmatrix}
   a^2 & 0   & 0   & 0 \\
   0   & b^2 & b^2 & b c \\
   0   & b^2 & b^2 & b c \\
   0   & b c & b c & 3 b^2 + c^2
  \end{bmatrix},
 \]
 which is redundant.
 Its compressed form is nonsingular since its determinant is $6 a^2 b^4 \ne 0$.
 Thus, $\PlHolant(f)$ is $\SHARPP$-hard by Corollary~\ref{cor:arity4:nonsingular_compressed_hard},
 so $\PlHolant(\mathcal{F})$ is also $\SHARPP$-hard.
\end{proof}

\section{An Improved Interpolation Technique} \label{sec:interpolation}

In the previous section, we used interpolation to show that counting the number of Eulerian orientations is $\SHARPP$-hard over planar 4-regular graphs.
Polynomial interpolation is a powerful tool in the study of counting problems that was initiated by Valiant~\cite{Val79a}.
In this section, we discuss a common interpolation method called the \emph{recursive unary construction} and obtain a tight characterization of when it succeeds.
The goal of this construction is to interpolate a unary signature and is based on work by Vadhan~\cite{Vad01} and further developed by others~\cite{CLX12, CLX11d, CK12}.

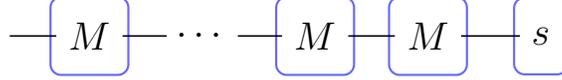
\begin{figure}[t]
 \centering
 \def\minLength{1.7cm}
 \begin{tikzpicture}[scale=\scale,transform shape,node distance=2.5cm,semithick]
  \node[external,white] (0)        {\Huge $M$};
  \node[external] (1) [right of=0] {\Huge $M$};
  \node[external] (2) [right of=1] {\Huge $\cdots$};
  \node[external] (3) [right of=2] {\Huge $M$};
  \node[external] (4) [right of=3] {\Huge $M$};
  \node[external] (5) [right of=4] {\Huge $s$};
  \node[external] (6) [right of=5] {};
  \path (0) edge (1)
        (1) edge (2)
        (2) edge (3)
        (3) edge (4)
        (4) edge (5);
  \begin{pgfonlayer}{background}
   \node[draw=\borderColor,thick,rounded corners,minimum height=\minLength,minimum width=\minLength,fit = (1)] {};
   \node[draw=\borderColor,thick,rounded corners,minimum height=\minLength,minimum width=\minLength,fit = (3)] {};
   \node[draw=\borderColor,thick,rounded corners,minimum height=\minLength,minimum width=\minLength,fit = (4)] {};
   \node[draw=\borderColor,thick,rounded corners,minimum height=\minLength,minimum width=       0cm,fit = (5)] {};
  \end{pgfonlayer}
 \end{tikzpicture}
 \caption{Recursive unary construction $(M,s)$.}
 \label{fig:gadget:recursive_unary_construction}
\end{figure}

There are two gadgets in the recursive unary construction:
a \emph{starter} gadget of arity~1 and a \emph{recursive} gadget of arity~2.
The signature of the starter gadget is represented by a two-dimensional column vectors $s$
and the signature of the recursive gadget is represented by a 2-by-2 matrix $M$.
The construction begins with the starter gadget and proceeds by connecting $k \ge 0$ recursive gadgets,
one at a time, to the only available edge (see Figure~\ref{fig:gadget:recursive_unary_construction}).
The signature of this gadget can be expressed as $M^k s$.
This construction is denoted by $(M,s)$.

The essential difficulty in using polynomial interpolation is constructing an infinite set of signatures that are pairwise linearly independent~\cite{CK12}.
The pairwise linear independence of signatures translates into distinct evaluation points for the polynomial being interpolated.
Thus, the essence of this interpolation technique can be stated as follows.

\begin{lemma}[Lemma~5.2 in~\cite{CLX11d}] \label{lem:pairwise_linearly_independent:2}
 Suppose $M \in \mathbb{C}^{2 \times 2}$ and $s \in \mathbb{C}^{2 \times 1}$.
 If the following three conditions are satisfied,
 \begin{enumerate}
  \item $\det(M) \ne 0$;
  \item $s$ is not a column eigenvector of $M$ (nor the zero vector);
  \item the ratio of the eigenvalues of $M$ is not a root of unity;
 \end{enumerate}
 then the vectors in the set $V = \{M^k s\}_{k \ge 0}$ are pairwise linearly independent.
\end{lemma}

Clearly the first condition is necessary.
The second condition is equivalent to $\det([s\ M s]) \ne 0$,
which is necessary since it checks the linear dependence of the first two vectors in $V$.

The recursive unary construction can be generalized to larger dimensions,
where the starter gadget has arity $d$ and the recursive gadget has arity $2 d$~\cite{KC10}.
In this generalized construction,
the starter gadget is represented by a column vector in $\mathbb{C}^{2^d}$ and the recursive gadget is represented by a matrix in $\mathbb{C}^{2^d \times 2^d}$.

For dimensions larger than one, the second condition in Lemma~\ref{lem:pairwise_linearly_independent:2} must be replaced by a stronger assumption,
such as ``\emph{$s$ is not orthogonal to any row eigenvector of $M$}''~\cite{CLX12}.
Previous work (Lemma~4.10 in~\cite{KC10_arXiv}, the full version of~\cite{KC10})
satisfied this stronger condition by showing that it follows from $\det([s\ M s\ \ldots\ M^{n-1} s]) \ne 0$.
For completeness, we show that these two conditions are equivalent.
We note that the use of $n$ instead of $2^d$ in the next two lemmas is not overly general.
Sometimes degeneracies or redundancies in the starter and recursive gadgets warrant the consideration of such cases.

\begin{lemma} \label{lem:2nd_condition_equivalence}
 Suppose $M \in \mathbb{C}^{n \times n}$ and $s \in \mathbb{C}^{n \times 1}$.
 Then $\det([s\ M s\ \ldots\ M^{n-1} s]) \ne 0$ iff $s$ is not orthogonal to any row eigenvector of $M$.
\end{lemma}

\begin{proof}
 Suppose $\det([s\ M s\ \ldots\ M^{n-1} s]) \ne 0$ and assume for a contradiction that $s$ is orthogonal to some row eigenvector $v$ of $M$ with eigenvalue $\lambda$.
 Then $v [s\ M s\ \ldots\ M^{n-1} s] = \mathbf{0}$ is the zero vector because $v M^i s = \lambda^i v s = 0$.
 Since $v \ne \mathbf{0}$, this a contradiction.
 
 Now suppose that $s$ is not orthogonal to any row eigenvector of $M$ and assume for a contradiction that $\det([s\ M s\ \ldots\ M^{n-1} s]) = 0$.
 Then there is a nonzero row vector $v$ such that $v [s\ M s\ \ldots\ M^{n-1} s] = \mathbf{0}$ is the zero vector.
 Consider the linear span $S$ by row vectors in the set $\{v, v M, \dotsc, v M^{n-1}\}$.
 We claim that $S$ is an invariant subspace of row vectors under the action of multiplication by $M$ from the right.

 By the Caylay-Hamilton theorem, $M$ satisfies its own characteristic polynomial, which is a monic polynomial of degree $n$.
 Thus, $M^n$ is a linear combination of $I_n, M, \dotsc, M^{n-1}$.
 This shows that for any $u \in S$, $u M$ still belongs to $S$.

 Therefore, there exists a $u \in S$ such that $u$ is a row eigenvector of $M$.
 By the definition of $S$, this $u$ is orthogonal to $s$, which is a contradiction.
\end{proof}

Another necessary condition, even for the $d$-dimensional case, is that $M$ has infinite order modulo a scalar.
Otherwise, $M^k = \beta I_n$ for some $k$ and any vector of the form $M^\ell s$ for $\ell \ge k$ is some multiple of a vector in the set $\{M^i s\}_{0 \le i < k}$.
We improve the $d$-dimensional version of Lemma~\ref{lem:pairwise_linearly_independent:2} by replacing the third condition with this necessary condition.

\begin{lemma} \label{lem:pairwise_linearly_independent:d}
 Suppose $M \in \mathbb{C}^{n \times n}$ and $s \in \mathbb{C}^{n \times 1}$.
 If the following three conditions are satisfied,
 \begin{enumerate}
  \item $\det(M) \ne 0$;
  \item $s$ is not orthogonal to any row eigenvector of $M$;
  \item $M$ has infinite order modulo a scalar;
 \end{enumerate}
 then the vectors in the set $V = \{M^k s\}_{k \ge 0}$ are pairwise linearly independent.
\end{lemma}

\begin{proof}
 Since $\det(M) \ne 0$, $M$ is nonsingular and the eigenvalues $\lambda_i$ of $M$, for $1 \le i \le n$, are nonzero.
 Let $M = P^{-1} J P$ be the Jordan decomposition of $M$ and let $p = P s \in \mathbb{C}^{n \times 1}$.
 Suppose for a contradiction that the vectors in $V$ are not pairwise linearly independent.
 This means that there exists integers $k > \ell \ge 0$ such that $M^k s = \beta M^\ell s$ for some nonzero complex value $\beta$.
 Let $t = k - \ell > 0$.
 Then we have $P^{-1} J^t P s = M^t s = \beta s$ and $J^t p = \beta p$.
 
 Suppose that $J$ contains some nontrivial Jordan block and consider the 2-by-2 submatrix in the bottom right corner of this block.
 From this portion of $J$,
 the two equations given by $J^t p = \beta p$ are $\lambda_i^t p_{i-1} + t \lambda_i^{t-1} p_i = \beta p_{i-1}$ and $\lambda_i^t p_i = \beta p_i$.
 Since $s$ is not orthogonal to any row eigenvector of $M$, $p_i \ne 0$.
 But then these equations imply that $t \lambda_i^{t-1} p_i = 0$, a contradiction.
 
 Otherwise, $J$ contains only trivial Jordan blocks.
 From $J^t p = \beta p$, we get the equations $\lambda_i p_i = \beta p_i$ for $1 \le i \le n$.
 Since $s$ is not orthogonal to any row eigenvector of $M$, $p_i \ne 0$ for $1 \le i \le n$.
 But then $M^t = \beta I_n$, which contradicts that fact that $M$ has infinite order modulo a scalar.
\end{proof}

With this lemma, we obtain a tight characterization for the success of interpolation by a recursive unary construction.
For example,
the construction using a recursive gadget with signature matrix $M = \left[\begin{smallmatrix} 1 & 1 \\ 0 & 1 \end{smallmatrix}\right]$
and a starter gadget with signature $s = \left[\begin{smallmatrix} 0 \\ 1 \end{smallmatrix}\right]$
is successful because $M$ and $s$ satisfy our conditions but do not satisfy previous sufficient conditions.

\begin{lemma} \label{lem:unary_recursive_construction}
 Let $\mathcal{F}$ be a set of signatures.
 If there exists a planar $\mathcal{F}$-gate with signature matrix $M \in \mathbb{C}^{2 \times 2}$ and
 a planar $\mathcal{F}$-gate with signature $s \in \mathbb{C}^{2 \times 1}$ satisfying the following conditions,
 \begin{enumerate}
  \item $\det(M) \ne 0$;
  \item $\det([s\ M s]) \ne 0$;
  \item $M$ has infinite order modulo a scalar;
 \end{enumerate}
 then $\PlHolant(\mathcal{F} \union \{[a,b]\}) \le_T \PlHolant(\mathcal{F})$ for any $a, b \in \mathbb{C}$.
\end{lemma}

\begin{proof}
 Consider an instance $\Omega = (G, \mathcal{F}, \pi)$ of $\PlHolant(\mathcal{F} \union \{[a,b]\})$.
 Let $V'$ be the subset of vertices assigned $[a,b]$ by $\pi$ and suppose that $|V'| = n$.
 We construct from $\Omega$ a sequence of instances $\Omega_k$ of $\PlHolant(\mathcal{F})$ indexed by $k \ge 1$.
 We obtain $\Omega_k$ from $\Omega$ by replacing each occurrence of $[a,b]$ with the recursive unary construction $(M,s)$
 in Figure~\ref{fig:gadget:recursive_unary_construction} containing $k$ copies of the recursive gadget.
 This recursive unary construction has the signature $[x_k, y_k] = M^k s$.

 By applying our assumptions to Lemmas~\ref{lem:2nd_condition_equivalence} 
 and~\ref{lem:pairwise_linearly_independent:d},
 we know that the signatures in the set $V = \{[x_k, y_k] \st 0 \le k \le n + 1\}$ 
 are pairwise linearly independent.
 In particular, at most one $y_k$ can be~0, so we may assume that $y_k \ne 0$ for $0 \le k \le n$, renaming variables if necessary.
 
 We stratify the assignments in $\Omega$ based on the assignment to $[a,b]$.
 Let $c_\ell$ be the sum over all assignments of products of evaluations at all $v \in V(G) - V'$
 such that exactly $\ell$ occurrences of $[a,b]$ have their incident edge assigned~0 (and $n - \ell$ have their incident edge assigned~1).
 Then
 \[
  \PlHolant_\Omega = \sum_{0 \le \ell \le n} a^\ell b^{n - \ell} c_\ell
 \]
 and the value of the Holant on $\Omega_k$, for $k \ge 1$, is
 \begin{align*}
  \PlHolant_{\Omega_k}
  &= \sum_{0 \le \ell \le n} x_k^\ell y_k^{n - \ell} c_\ell\\
  &= y_k^n \sum_{0 \le \ell \le n} \left(\frac{x_k}{y_k}\right)^\ell c_\ell.
 \end{align*}
 The coefficient matrix of this linear system is Vandermonde.
 Since the signatures in $V$ are pairwise linearly independent,
 the ratios $x_k / y_k$ are distinct (and well-defined since $y_k \ne 0$),
 which means that the Vandermonde matrix has full rank.
 Therefore, we can solve the linear system for the unknown $c_\ell$'s and obtain the value of $\Holant_\Omega$.
\end{proof}

The first two conditions of Lemma~\ref{lem:unary_recursive_construction} are easy to check.
The third condition holds in one of these two cases:
either the eigenvalues are the same but $M$ is not a multiple of the identity matrix,
or the eigenvalues are different but their ratio is not a root of unity.

Our refined conditions work well with the anti-gadget technique~\cite{CKW12}.
The power of this lemma is that when the third condition fails to hold,
there exists an integer $k$ such that $M^k = I_2$, where $I_2$ is the 2-by-2 identity matrix.
Therefore we can construct $M^{k-1} = M^{-1}$ and use this in other gadget constructions.

The anti-gadget technique is used in combination with Lemma~\ref{lem:unary_recursive_construction}
to give a succinct proof of Lemma~\ref{lem:arity4:double_root_interpolation}.
The construction in this proof is actually not a recursive unary construction,
but a recursive binary construction.
However, degeneracies in the starter and recursive gadgets permit analysis 
equivalent to that of the recursive unary construction.
We also use the anti-gadget technique and the power of Lemma~\ref{lem:unary_recursive_construction}
(via Lemma~\ref{lem:unary_interpolation:EQ-hat}) in the proof of 
Theorem~\ref{thm:pinning} to handle a difficult case.

\section{Pl-Holant Dichotomy for a Symmetric Arity~4 Signature} \label{sec:arity-4}

With Corollary~\ref{cor:arity4:nonsingular_compressed_hard:symmetric} in hand,
only one obstacle remains in proving a dichotomy for a symmetric arity~4 signature in the Pl-Holant framework:
the case $[v,1,0,0,0]$ when $v$ is different from~0.
Over the next two lemmas, we prove that this problem is $\SHARPP$-hard by a reduction from $\PlHolant([v,1,0,0])$.
These problems are the weighted versions of counting matchings over planar $k$-regular graphs for $k = 4$ and $k = 3$ respectively.

In the first lemma,
we show how to use either the anti-gadget technique from~\cite{CKW12}
or interpolation by our tight characterization of the recursive unary construction from Section~\ref{sec:interpolation} to effectively obtain $[1,0,0]$.

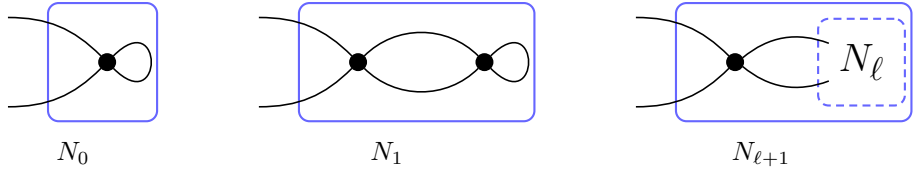
\begin{figure}[t]
 \centering
 \captionsetup[subfigure]{labelformat=empty}
 \subcaptionbox{$N_0$}{
  \begin{tikzpicture}[scale=\scale,transform shape,node distance=\nodeDist,semithick]
   \node[external] (0)                    {};
   \node[external] (1) [right       of=0] {};
   \node[internal] (2) [below right of=1] {};
   \node[external] (3) [below left  of=2] {};
   \node[external] (4) [left        of=3] {};
   \path (0) edge[out=   0, in=135] (2)
         (2) edge[out=-135, in=  0] (4)
         (2) edge[out=45,   in=-45, looseness=15] coordinate (c1) (2);
   \begin{pgfonlayer}{background}
    \node[draw=\borderColor,thick,rounded corners,fit = (1) (3) (c1)] {};
   \end{pgfonlayer}
  \end{tikzpicture}}
 \qquad
 \subcaptionbox{$N_1$}{
  \begin{tikzpicture}[scale=\scale,transform shape,node distance=\nodeDist,semithick]
   \node[external]  (0)                    {};
   \node[external]  (1) [right       of=0] {};
   \node[internal]  (2) [below right of=1] {};
   \node[external]  (3) [below left  of=2] {};
   \node[external]  (4) [left        of=3] {};
   \node[external]  (5) [right       of=2] {};
   \node[internal]  (6) [right       of=5] {};
   \path (0) edge[out=   0, in= 135] (2)
         (2) edge[out=-135, in=   0] (4)
             edge[out=  45, in= 135] (6)
             edge[out= -45, in=-135] (6)
         (6) edge[out=  45, in= -45, looseness=15] coordinate (c1) (6);
   \begin{pgfonlayer}{background}
    \node[draw=\borderColor,thick,rounded corners,fit = (1) (3) (c1)] {};
   \end{pgfonlayer}
  \end{tikzpicture}}
 \qquad
 \subcaptionbox{$N_{\ell+1}$}{
  \begin{tikzpicture}[scale=\scale,transform shape,node distance=\nodeDist,semithick]
   \node[external]  (0)                    {};
   \node[external]  (1) [right       of=0] {};
   \node[internal]  (2) [below right of=1] {};
   \node[external]  (3) [below left  of=2] {};
   \node[external]  (4) [left        of=3] {};
   \node[external]  (5) [right       of=2] {};
   \node[external]  (6) [right       of=5] {\Huge $N_\ell$};
   \path (0) edge[out=   0, in=135] (2)
         (2) edge[out=-135, in=  0] (4)
             edge[bend left]        (6)
             edge[bend right]       (6);
   \begin{pgfonlayer}{background}
    \node[draw=\borderColor,thick,densely dashed,rounded corners,fit = (6)] {};
    \node[draw=\borderColor,thick,rounded corners,fit = (1) (3) (c1)] {}; %(c1) is defined in the previous tikzpicture, which seems to be allowed
   \end{pgfonlayer}
  \end{tikzpicture}}
 \caption{Binary recursive construction to interpolate $[1,0,0]$. The vertices are assigned $[v,1,0,0,0]$.}
 \label{fig:gadget:arity4:interpolate_100}
\end{figure}

\begin{lemma} \label{lem:arity4:double_root_interpolation}
 For any $v \in \mathbb{C}$ and signature set $\mathcal{F}$ containing $[v,1,0,0,0]$,
 \[\PlHolant(\mathcal{F} \union \{[1,0,0]\}) \le_T \PlHolant(\mathcal{F}).\]
\end{lemma}

\begin{proof}
 Consider the gadget construction in Figure~\ref{fig:gadget:arity4:interpolate_100}.
 For $k \ge 0$, the signature of $N_k$ is of the form $[a_k, b_k, 0]$, and $N_0 = [v,1,0]$.
 Since $N_k$ is symmetric and always ends with~0,
 we can analyze this construction as though it were a recursive unary construction.
 Let $s_k = \transpose{[a_k, b_k]}$, so $s_0 = \transpose{[v,1]}$.
 It is clear that $s_k = M^k s_0$,
 where $M = \left[\begin{smallmatrix} v & 2 \\ 1 & 0 \end{smallmatrix}\right]$.
 
 Since $\det(M) = -2$, $M$ is nonsingular.
 If $M$ has finite order modulo a scalar,
 then $M^\ell = \beta I_2$ for some positive integer $\ell$ and some nonzero complex value $\beta$.
 Thus, the signature of $N_{\ell - 1}$, which contains the anti-gadget of $M$, is $M^{\ell - 1} s_0 = \beta M^{-1} s_0 = \beta \transpose{[1,0]}$.
 After normalizing, we directly realize $[1,0,0]$.
 
 Now assume that $M$ has infinite order modulo a scalar.
 Since $\det([s_0\ M s_0]) = -2$, we can interpolate any signature of the form $[x,y,0]$ by Lemma~\ref{lem:unary_recursive_construction}, including $[1,0,0]$.
\end{proof}

For the next lemma, we use a well-known and easy generalization of a classic result of Petersen~\cite{Pet91}.
Petersen's theorem considers 3-regular, bridgeless, simple graphs
(i.e.~graphs without self-loops or parallel edges) and concludes that there exists a perfect matching.
The same conclusion holds even if the graphs are not simple.
We provide a proof for completeness.

\begin{theorem} \label{thm:arity4:PM_in_bridgeless}
 Any 3-regular bridgeless graph $G$ has a perfect matching.
\end{theorem}

\begin{proof}
 We may assume that $G$ is connected.
 If $G$ has a vertex $v$ with a self-loop, then the other edge of $v$ is a bridge since $G$ is 3-regular, which is a contradiction.
 If there exists some pair of vertices of $G$ joined by exactly three parallel edges, then $G$ has only these two vertices since it is connected and the theorem holds.

 In the remaining case, there exists some pair of vertices joined by exactly two parallel edges.
 We build a new graph $G'$ without any parallel edges.
 For vertices $u$ and $v$ joined by exactly two parallel edges,
 we remove these two parallel edges and introduce two new vertices $w_1$ and $w_2$.
 We also introduce the new edges $(u, w_1)$, $(u, w_2)$, $(v, w_1)$, $(v, w_2)$, and $(w_1, w_2)$.
 Then $G'$ is a 3-regular, bridgeless, simple graph.
 
 By Petersen's theorem, $G'$ has a perfect matching $P'$.
 Now we construct a perfect matching $P$ in $G$ using $P'$.
 We put any edge in both $G$ and $P'$ into $P$.
 If $u$ is matched by a new edge in $G'$, then $v$ must be matched by a new edge in $G'$ as well and we put the edge $(u, v)$ into $P$.
 If $u$ and $v$ are not matched by a new edge,
 then we do not add anything to $P$.
 It is easy to see that $P$ is a perfect matching in $G$.
\end{proof}

We use this result to show the existence of what we call a \emph{planar pairing} for any planar 3-regular graph,
which we use in our proof of $\SHARPP$-hardness.

\begin{definition}[Planar pairing]
 A \emph{planar pairing} in a graph $G = (V, E)$ is a set of edges $P \subset V \times V$
 such that $P$ is a perfect matching in the graph $(V, V \times V)$,
 and the graph $(V, E \union P)$ is planar.
\end{definition}

Obviously, a perfect matching in the original graph is a planar pairing.

\begin{lemma} \label{lem:planar:pairing}
 For any planar 3-regular graph $G$,
 there exists a planar pairing that can be computed in polynomial time.
\end{lemma}

\begin{proof}
 We efficiently find a planar pairing in $G$ by induction on the number of vertices in $G$.
 Since $G$ is a 3-regular graph, it must have an even number of vertices.
 If there are no vertices in $G$, then there is nothing to do.
 Suppose that $G$ has $n = 2 k$ vertices and that we can efficiently find a planar pairing in graphs containing fewer vertices.
 If $G$ is not connected, then we can already apply our inductive hypothesis on each connected component of $G$.
 The union of planar pairings in each connected component of $G$ is a planar pairing in $G$, so we are done.
 Otherwise assume that $G$ is connected.

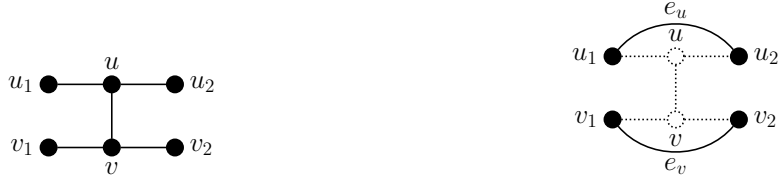
\begin{figure}[t]
 \centering
 \tikzstyle{removed} = [densely dotted, draw, black, fill=white, shape=circle]
 \def\bridgeWidth{6.6cm}
 \captionsetup[subfigure]{width=\bridgeWidth}
 \subcaptionbox{The neighborhood around $u$ and $v$ in $G$.\label{subfig:uv_neighborgood}}[\bridgeWidth]{
  \begin{tikzpicture}[scale=\scale,transform shape,node distance=\nodeDist,semithick,font=\LARGE]
   \node[internal] (0) [label=left:$u_1$]              {};
   \node[internal] (1) [below of=0, label=left:$v_1$]  {};
   \node[internal] (2) [right of=0, label=above:$u$]   {};
   \node[internal] (3) [right of=1, label=below:$v$]   {};
   \node[internal] (4) [right of=2, label=right:$u_2$] {};
   \node[internal] (5) [right of=3, label=right:$v_2$] {};
   \path (0) edge (2)
         (1) edge (3)
         (2) edge (3)
             edge (4)
         (3) edge (5);
  \end{tikzpicture}}
 \qquad
 \subcaptionbox{The same neighborhood in $H$.\label{subfig:neighborhood_without_uv}}[\bridgeWidth]{
  \begin{tikzpicture}[scale=\scale,transform shape,node distance=\nodeDist,semithick,font=\LARGE]
   \node[internal] (0) [            label=left:$u_1$]  {};
   \node[internal] (1) [below of=0, label=left:$v_1$]  {};
   \node[removed]  (2) [right of=0, label=above:$u$]   {};
   \node[removed]  (3) [right of=1, label=below:$v$]   {};
   \node[internal] (4) [right of=2, label=right:$u_2$] {};
   \node[internal] (5) [right of=3, label=right:$v_2$] {};
   \path (0) edge[densely dotted] (2)
         (1) edge[densely dotted] (3)
         (2) edge[densely dotted] (3)
             edge[densely dotted] (4)
         (3) edge[densely dotted] (5);
   \path (0) edge[out= 50, in= 130] node[label=above:$e_u$] {} (4)
         (1) edge[out=-50, in=-130] node[label=below:$e_v$] {} (5);
  \end{tikzpicture}}
 \caption{The neighborhood around $u$ and $v$ both before and after they are removed.}
 \label{fig:bridge_neighborhood}
\end{figure}

 Suppose that $G$ contains a bridge $(u,v)$.
 Let the three (though not necessarily distinct) neighbors of $u$ be $v$, $u_1$, and $u_2$,
 and let the three (though not necessarily distinct) neighbors of $v$ be $u$, $v_1$, and $v_2$ (see Figure~\ref{subfig:uv_neighborgood}).
 Furthermore, let $H_u$ be the connected component in $G - \{(u,v)\}$ containing $u$
 and let $H_v$ be the connected component in $G - \{(u,v)\}$ containing $v$.
 Consider the induced subgraph $H_u'$ of $H_u$ after adding the edge $e_u = (u_1, u_2)$ (which might be a self-loop) and removing $u$.
 Similarly, consider the induced subgraph $H_v'$ of $H_v$ after adding the edge $e_v = (v_1, v_2)$ (which might be a self-loop) and removing $v$.
 Both $H_u'$ and $H_v'$ are 3-regular graphs and their disjoint union gives a graph $H'$ with $n-2 = 2(k-1)$ vertices (see Figure~\ref{subfig:neighborhood_without_uv}).
 
 By induction on both $H_u'$ and $H_v'$,
 we have planar pairings $P_u$ and $P_v$ in $H_u'$ and $H_v'$ respectively.
 Let $H''$ be the graph $H'$ including the edges $P_u \union P_v$.
 If $H''$ contains both $e_u$ and $e_v$,
 then embed $H''$ in the plane so that both $e_u$ and $e_v$ are adjacent to the outer face.
 Then the graph $G$ including the edges $P_u \union P_v$ is also planar,
 so $P_u \union P_v \union \{(u,v)\}$ is a planar pairing in $G$.
 
 Otherwise, $G$ is bridgeless.
 Then by Theorem~\ref{thm:arity4:PM_in_bridgeless}, $G$ has a perfect matching, which is also a planar pairing in $G$.
 Since a perfect matching can be found in polynomial time by Edmond's blossom algorithm~\cite{Edm65},
 the whole procedure is in polynomial time.
\end{proof}

After publishing a preliminary version of this paper~\cite{GW13},
we learned that Cai and Kowalczyk had previously used the planar pairing technique to show that
counting the number of vertex covers over $k$-regular graphs is $\SHARPP$-hard for even $k \ge 4$ (see the proof of Lemma~15 in~\cite{CK13}).
Their algorithm to find a planar pairing starts by taking a spanning tree and then pairing up the vertices on this tree,
which is simpler than our approach.
We believe that it is worth emphasizing the importance of this technique.
Most gadget constructions in hardness proofs for Holant problems are local but the planar pairing technique is a global argument,
which permits reductions that are not otherwise possible.

Now we use the planar pairing technique to show the following.

\begin{lemma} \label{lem:arity4:double_root}
 If $v \in \mathbb{C} - \{0\}$,
 then $\PlHolant([v,1,0,0,0])$ is $\SHARPP$-hard.
\end{lemma}

\begin{proof}
 We reduce from $\PlHolant([v,1,0,0])$ to $\PlHolant([v,1,0,0,0])$.
 Since $\PlHolant([v,1,0,0])$ is $\SHARPP$-hard when $v \neq 0$ by Theorem~\ref{thm:arity3:singleton},
 this shows that $\PlHolant([v,1,0,0,0])$ is also $\SHARPP$-hard when $v \neq 0$.
 
 An instance of $\PlHolant([v,1,0,0])$ is a signature grid $\Omega$ with underlying graph $G = (V, E)$ that is planar and 3-regular.
 By Lemma~\ref{lem:planar:pairing}, there exists a planar pairing $P$ in $G$ and it can be found in polynomial time.
 Then the graph $G' = (V, E \union P)$ is planar and 4-regular.
 We assign $[v,1,0,0,0]$ to every vertex in $G'$.
 By Lemma~\ref{lem:arity4:double_root_interpolation}, we can assume that we have $[1,0,0]$.
 We replace each edge in $P$ with a path of length~2 to form a graph $G''$ and assign $[1,0,0] = [1,0]^{\otimes 2}$ to each of the new vertices.
 Then the signature grid $\Omega''$ with underlying graph $G''$ has the same Holant value as the original signature grid $\Omega$.
\end{proof}

Note that our proof of Lemma~\ref{lem:arity4:double_root} reduces $\PlHolant([v,1,0,0])$ to $\PlHolant([v,1,0,0,0])$ for all $v \in \mathbb{C}$.
Neither Lemma~\ref{lem:arity4:double_root_interpolation} nor Lemma~\ref{lem:arity4:double_root} ever considers the value of $v$.
This is consistent because both signatures are in $\mathscr{M}$ when $v = 0$, thus tractable,
and both signatures are $\SHARPP$-hard when $v$ is different from zero.

Now we are ready to prove our Pl-Holant dichotomy for a symmetric arity~4 signature.
A signature is called \emph{vanishing} if the Holant of any signature grid using only that signature is zero~\cite{CGW13}.

\begin{theorem} \label{thm:arity4:singleton}
 If $f$ is a non-degenerate, symmetric, complex-valued signature of arity 4 in Boolean variables,
 then $\PlHolant(f)$ is $\SHARPP$-hard unless
 $f$ is $\mathscr{A}$-transformable or $\mathscr{P}$-transformable or vanishing or $\mathscr{M}$-transformable,
 in which case the problem is in $\P$.
\end{theorem}

\begin{proof}
 Let $f = [f_0, f_1, f_2, f_3, f_4]$.
 If there do not exist $a, b, c \in \mathbb{C}$, not all zero,
 such that for all $k \in \{0,1,2\}$, $a f_k + b f_{k+1} + c f_{k+2} = 0$,
 then $\PlHolant(f)$ is $\SHARPP$-hard by Corollary~\ref{cor:arity4:nonsingular_compressed_hard:symmetric}.
 Otherwise, there do exist such $a, b, c$.
 If $a = c = 0$, then $b \ne 0$, so $f_1 = f_2 = f_3 = 0$.
 In this case, $f \in \mathscr{P}$ is a generalized equality signature, so $f$ is $\mathscr{P}$-transformable.
 Now suppose $a$ and $c$ are not both 0.
 If $b^2 - 4 a c \ne 0$,
 then $f_k = \alpha_1^{4-k} \alpha_2^k + \beta_1^{4-k} \beta_2^k$,
 where $\alpha_1 \beta_2 - \alpha_2 \beta_1 \ne 0$.
 A holographic transformation by $\left[\begin{smallmatrix} \alpha_1 & \beta_1 \\ \alpha_2 & \beta_2 \end{smallmatrix}\right]$ transforms $f$ to $=_4$
 and we can use Theorem~\ref{thm:degree_prescribed_homomorphism} to show that $f$ is either $\mathscr{A}$-, $\mathscr{P}$-, or $\mathscr{M}$-transformable
 unless $\PlHolant(f)$ is $\SHARPP$-hard.
 Otherwise, $b^2 - 4 a c = 0$ and there are two cases.
 In the first,  for any $0 \le k \le 2$, $f_k = c k \alpha^{k-1} + d \alpha^k$, where $c \ne 0$.
 In the second, for any $0 \le k \le 2$, $f_k = c (4-k) \alpha^{3-k} + d \alpha^{4-k}$, where $c \ne 0$.
 These cases map between each other under a holographic transformation by $\left[\begin{smallmatrix} 0 & 1 \\ 1 & 0\end{smallmatrix}\right]$,
 so assume that we are in the first case.
 If $\alpha = \pm i$, then $f$ is vanishing.
 Otherwise, a further holographic transformation by $\frac{1}{\sqrt{1 + \alpha^2}} \left[\begin{smallmatrix} 1 & \alpha \\ \alpha & -1 \end{smallmatrix}\right]$
 transforms $f$ to $\hat{f} = [v,1,0,0,0]$ for some $v \in \mathbb{C}$ after normalizing the second entry.
 (See Appendix~B in~\cite{CGW13} for details.)
 If $v = 0$, then the problem is counting perfect matchings over planar 4-regular graphs,
 so $\hat{f} \in \mathscr{M}$ and $f$ is $\mathscr{M}$-transformable.
 Otherwise, $v \ne 0$ and we are done by Lemma~\ref{lem:arity4:double_root}.
\end{proof}

\section{Domain Pairing} \label{sec:pairing}

Now we turn our attention to our main result, a dichotomy for the $\PlCSP$ framework.
In this section, we discuss a technique called \emph{domain pairing},
which pairs input variables to simulate a problem on a domain of size four and then reduces a problem in the Boolean domain to it.
As explained in the introduction, we work in the Hadamard basis instead of the standard basis.
The goal then becomes a dichotomy for $\PlHolant(\mathcal{F} \union \widehat{\EQ})$.

In~\cite{CGW13}, a simple interpolation lemma for non-degenerate, generalized equality signatures of arity at least~3 was proved.
Although the lemma was only for general graphs, it was mentioned that it also holds for planar graphs.

\begin{lemma}[Lemma~A.2 in~\cite{CGW13}] \label{lem:=4_from_generalized_=}
 Let $a, b \in \mathbb{C}$.
 If $a b \ne 0$, then for any set $\mathcal{F}$ of complex-weighted signatures containing $[a, 0, \dotsc, 0, b]$ of arity at least~3,
 \[\PlHolant(\mathcal{F} \union \{{=}_4\}) \le_T \PlHolant(\mathcal{F}).\]
\end{lemma}

By a simple parity argument, gadgets constructed with signatures of even arity can only realize other signatures of even arity.
In particular, this means that $=_4$ cannot by itself be used to construct $=_3$.
Nevertheless, there is a clever argument that can realize $=_3$ using $=_4$.
The catch is the domain changes from individual elements to pairs of elements.
Thus, we call this reduction technique \emph{domain pairing}.
This technique was first used in the proof of Lemma~III.2 in~\cite{CLX10} with real weights.
It was also used in the proof of Lemma~4.6 in~\cite{GLV13} in the parity case
and in Lemma~IV.5 in~\cite{HL12} with real weights as well as grouping more than just two domain elements.
We prove a generalization of the domain pairing lemma for complex weights.

\begin{lemma}[Domain pairing] \label{lem:domain_pairing:strong}
 Let $a, b, x, y \in \mathbb{C}$.
 If $a b y \ne 0$ and $x^2 \ne y^2$,
 then for any set $\mathcal{F}$ of complex-valued symmetric signatures containing $[x,0,y,0]$ and $[a,0,\dotsc,0,b]$ of arity at least~3,
 $\PlHolant(\mathcal{F} \union \widehat{\EQ})$ is $\SHARPP$-hard.
\end{lemma}

\begin{proof}
 We reduce from $\plholant{[x,y,y]}{\EQ}$ to $\PlHolant(\mathcal{F} \union \widehat{\EQ})$.
 Since $\plholant{[x,y,y]}{\EQ}$ is $\SHARPP$-hard when $y \ne 0$ and $x^2 \ne y^2$ by Theorem~\ref{thm:Pl-Graph_Homomorphism},
 this shows that $\PlHolant(\mathcal{F} \union \widehat{\EQ})$ is also $\SHARPP$-hard.
 
 An instance of $\plholant{[x,y,y]}{\EQ}$ is a signature grid $\Omega$ with underlying graph $G = (U, V, E)$.
 In addition to $G$ being bipartite and planar, every vertex in $U$ has degree~$2$.
 We replace every vertex in $V$ of degree $k$ (which is assigned ${=}_k \in \EQ$) with a vertex of degree $2k$,
 and bundle two adjacent variables to form $k$ bundles of~2 edges each.
 The $k$ bundles correspond to the $k$ incident edges of the original vertex with degree $k$.
 By Lemma~\ref{lem:=4_from_generalized_=}, we have $=_4$,
 which we use to construct $=_{2k}$ for any $k$.
 Then we assign $=_{2k}$ to the new vertices of degree $2k$.
 
 If the inputs to these equality signatures are restricted to $\{(0,0), (1,1)\}$ on each bundle,
 then these equality signatures take value~1 on $((0,0), \dotsc, (0,0))$ and $((1,1), \dotsc, (1,1))$ and take value~0 elsewhere.
 Thus, if we restrict the domain to $\{(0,0), (1,1)\}$, it is the equality signature $=_k$.

\begin{figure}[t]
 \centering
 \begin{tikzpicture}[scale=\scale,transform shape,node distance=\nodeDist,semithick,font=\LARGE]
  \node[external]  (0) [            label=left:$a_1$]  {};
  \node[external]  (1) [below of=0, label=left:$a_2$]  {};
  \node[internal]  (2) [right of=0]                    {};
  \node[internal]  (3) [right of=1]                    {};
  \node[external]  (4) [right of=2, label=right:$b_1$] {};
  \node[external]  (5) [right of=3, label=right:$b_2$] {};
  \path (0) edge node[very near end]   (n1) {} (2)
        (1) edge                               (3)
        (2) edge node[label=right:$c$]      {} (3)
            edge                               (4)
        (3) edge node[very near start] (n2) {} (5);
  \begin{pgfonlayer}{background}
   \node[draw=\borderColor,thick,rounded corners,inner sep=2mm,fit = (n1) (n2)] {};
  \end{pgfonlayer}
 \end{tikzpicture}
 \caption{Gadget designed for the paired domain. One vertex is assigned $[1,0,1,0]$ and the other is assigned $[x,0,y,0]$.}
 \label{fig:gadget:domain_pairing}
\end{figure}
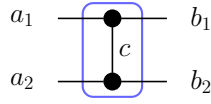
 
 To simulate $[x,y,y]$,
 we connect $f = [x,0,y,0]$ to $g = [1,0,1,0] \in \widehat{\EQ}$ by a single edge as shown in Figure~\ref{fig:gadget:domain_pairing} to form a gadget with signature
 \[h(a_1, a_2, b_1, b_2) = \sum_{c=0,1} f(a_1, b_1, c) g(a_2, b_2, c).\]
 We replace every (degree~2) vertex in $U$ (which is assigned $[x,y,y]$) by a degree~4 vertex assigned $h$,
 where the variables of $h$ are bundled as $(a_1, a_2)$ and $(b_1, b_2)$.
 
 The vertices in this new graph $G'$ are connected as in the original graph $G$,
 except that every original edge is replaced by two edges that connect to the same side of the gadget in Figure~\ref{fig:gadget:domain_pairing}.
 Notice that $h$ is only connected by $(a_1, a_2)$ and $(b_1, b_2)$ to some bundle of two incident edges of an equality signature.
 Since this equality signature enforces that the value on each bundle is either $(0,0)$ or $(1,1)$,
 we only need to consider the restriction of $h$ to the domain $\{(0,0), (1,1)\}$.
 On this domain, $h = [x,y,y]$ is a \emph{symmetric} signature of arity~2.
 Therefore, the signature grid $\Omega'$ with underlying graph $G'$ has the same Holant value as the original signature grid $\Omega$.
\end{proof}

There are two scenarios that lead to Lemma~\ref{lem:domain_pairing:strong}.
The proof of the first is immediate.

\begin{corollary} \label{cor:domain_pairing}
 Let $a, b \in \mathbb{C}$.
 If $a b x y \ne 0$ and $x^4 \ne y^4$,
 then for any set $\mathcal{F}$ of complex-weighted symmetric signatures containing $[x,0,y]$ and $[a,0,\dotsc,0,b]$ of arity at least~3,
 $\PlHolant(\mathcal{F} \union \widehat{\EQ})$ is $\SHARPP$-hard.
\end{corollary}

\begin{proof}
 Connect three copies of $[x,0,y]$ to $[1,0,1,0]$, with one on each edge, to get $x [x^2,0,y^2,0]$ and apply Lemma~\ref{lem:domain_pairing:strong}.
\end{proof}

The second scenario that leads to Lemma~\ref{lem:domain_pairing:strong} is Lemma~\ref{lem:domain_pairing:weak}.
The proof of Lemma~\ref{lem:domain_pairing:weak} applies Corollary~\ref{cor:domain_pairing} after interpolating a unary signature in one of two ways.
The next lemma considers one of those ways.

\begin{lemma} \label{lem:unary_interpolation:EQ-hat}
 Suppose $x \in \mathbb{C}$ and let $f = [1,x,1]$.
 If $x \not\in \{0, \pm 1\}$ and $M_f$ has infinite order modulo a scalar,
 then for any set $\mathcal{F}$ of complex-weighted symmetric signatures containing $f$ and for any $a, b \in \mathbb{C}$, we have
 \[\PlHolant(\mathcal{F} \union \{[a,b]\} \union \widehat{\EQ}) \le_T \PlHolant(\mathcal{F} \union \widehat{\EQ}).\]
\end{lemma}

\begin{proof}
 Consider the recursive unary construction $(M_f,s)$, where $s = \transpose{[1,0]}$.
 The determinant of $M_f$ is $1 - x^2 \ne 0$.
 The determinant of $[s\ M_f s]$ is $x \ne 0$.
 By assumption, $M_f$ has infinite order modulo a scalar.
 Therefore, we can interpolate any unary signature by Lemma~\ref{lem:unary_recursive_construction}.
\end{proof}

\begin{lemma} \label{lem:domain_pairing:weak}
 Let $a, b \in \mathbb{C}$.
 If $a b \ne 0$ and $a^4 \ne b^4$,
 then for any set $\mathcal{F}$ of complex-weighted symmetric signatures containing $f = [a,0,\dotsc,0,b]$ of arity at least~3,
 $\PlHolant(\mathcal{F} \union \widehat{\EQ})$ is $\SHARPP$-hard.
\end{lemma}

\begin{proof}
 Since $a \ne 0$, we normalize $f$ to $[1,0,\dotsc,0,x]$, where $x \ne 0$ and $x^4 \ne 1$.
 If the arity of $f$ is even, then after some number of self-loops, we have $[1,0,x]$ and are done by Corollary~\ref{cor:domain_pairing}.
 Otherwise, the arity of $f$ is odd.
 After some number of self-loops, we have $g = [1,0,0,x]$.
 If we had the signature $[1,1]$, then we could connect this to $g$ to get $[1,0,x]$ and be done by Corollary~\ref{cor:domain_pairing}.
 We now show how to interpolate $[1,1]$ in one of two ways.
 
 Suppose $\Re(x)$, the real part of $x$, is not $0$.
 One more self-loop on $g$ gives $[1,x]$ and connecting this to $[1,0,1,0]$ gives $h = [1,x,1]$.
 The eigenvalues of $M_h$ are $\lambda_\pm = 1 \pm x$.
 Since $\Re(x) \ne 0$ iff $|\frac{\lambda_+}{\lambda_-}| \ne 1$, the ratio of the eigenvalues is not a root of unity, so $M_h$ has infinite order modulo a scalar.
 Therefore, we can interpolate $[1,1]$ by Lemma~\ref{lem:unary_interpolation:EQ-hat}.
 
 Otherwise, $\Re(x) = 0$ but $x$ is not a root of unity since $x \ne \pm i$.
 Connecting $[1,x]$ to $g$ gives $h = [1,0,x^2]$.
 Consider the recursive unary construction $(M_h,s)$, where $s = \transpose{[1,x]}$.
 The determinant of $M_h$ is $x^2 \ne 0$, so its eigenvalues are nonzero.
 Also, the determinant of $[s\ M_h s]$ is $x (x^2 - 1) \ne 0$.
 The ratio of the eigenvalues of $M_h$ is $x^2$, which is not a root of unity since $x$ is not a root of unity.
 Therefore $M_h$ has infinite order modulo a scalar and we can interpolate $[1,1]$ by Lemma~\ref{lem:unary_recursive_construction}.
\end{proof}

\section{Mixing of Tractable Signatures} \label{sec:mixing}

In this section, we determine which tractable signatures combine to give $\SHARPP$-hardness.
To help understand the various cases considered in the lemmas,
there is a Venn diagram of the signatures in $\mathscr{A}$, $\widehat{\mathscr{P}}$,
and $\mathscr{M}$ in Figure~\ref{fig:venn_diagram} of Appendix~\ref{sec:venn_diagram}.

The first two lemmas consider the case when one of the signatures has arity one.

\begin{lemma} \label{lem:mixing:unaries:AP}
 Suppose $f \in \mathscr{A} - \widehat{\mathscr{P}}$.
 If $a b \ne 0$ and $a^4 \ne b^4$,
 then for any set $\mathcal{F}$ of complex-weighted symmetric signatures containing $f$ and $[a,b]$,
 $\PlHolant(\mathcal{F} \union \widehat{\EQ})$ is $\SHARPP$-hard.
\end{lemma}

\begin{proof}
 Up to a nonzero scalar, the possibilities for $f$ are
 \begin{itemize}
  \item $[1, 0, \pm i]$;
  \item $[1, 0, \ldots, 0, x]$ of arity at least~3 with $x^4 = 1$;
  \item $[1, \pm 1, -1, \mp 1, 1, \pm 1, -1, \mp 1, \dotsc]$ of arity at least~2;
  \item $[1,  0, -1,  0, 1,  0, -1,  0, \dotsc, 0 \text{ or } 1 \text{ or } (-1)]$ of arity at least~3;
  \item $[0,  1,  0, -1, 0,  1,  0, -1, \dotsc, 0 \text{ or } 1 \text{ or } (-1)]$ of arity at least~3.
 \end{itemize}
 We handle these cases below.
 
 \begin{enumerate}
  \item Suppose $f = [1, 0, \pm i]$.
  Connecting $[a,b]$ to $[1,0,1,0]$ gives $[a,b,a]$ and connecting two copies of $[1, 0, \pm i]$ to $[a,b,a]$, one on each edge, gives $g = [a, \pm i b, -a]$.
  Since $a^4 \ne b^4$, $\plholant{g}{\widehat{\EQ}}$ is $\SHARPP$-hard by Theorem~\ref{thm:Pl-Graph_Homomorphism},
  so $\PlHolant(\mathcal{F} \union \widehat{\EQ})$ is also $\SHARPP$-hard.
  
  \item Suppose $f = [1,0,\dotsc,0,x]$ of arity at least~3 with $x^4 = 1$.
  Connecting $[a,b]$ to $f$ gives $g = [a, 0, \dotsc, 0, b x]$ of arity at least~2.
  Note that $(b x)^4 = b^4 \ne a^4$.
  If the arity of $g$ is exactly~2, then $\PlHolant(\{f,g\} \union \widehat{\EQ})$ is $\SHARPP$-hard by Corollary~\ref{cor:domain_pairing},
  so $\PlHolant(\mathcal{F} \union \widehat{\EQ})$ is also $\SHARPP$-hard.
  Otherwise, the arity of $g$ is at least~3 and $\PlHolant(\{g\} \union \widehat{\EQ})$ is $\SHARPP$-hard by Lemma~\ref{lem:domain_pairing:weak},
  so $\PlHolant(\mathcal{F} \union \widehat{\EQ})$ is also $\SHARPP$-hard.
  
  \item Suppose $f = [1, \pm 1, -1, \dotsc]$ of arity at least~2.
  Connecting some number of $[1,0]$ gives $[1, \pm 1, -1]$ of arity exactly~2.
  Connecting $[a,b]$ to $[1,0,1,0]$ gives $[a,b,a]$ and connecting two copies of $[a,b,a]$ to $[1, \pm 1, -1]$,
  one on each edge, gives $g = [a^2 \pm 2 a b - b^2, \pm (a^2 + b^2), -a^2 \pm 2 a b + b^2]$.
  This is easily verified by
  \[
   \begin{bmatrix} a &     b \\     b &  a \end{bmatrix}
   \begin{bmatrix} 1 & \pm 1 \\ \pm 1 & -1 \end{bmatrix}
   \begin{bmatrix} a &     b \\     b &  a \end{bmatrix}
   =
   \begin{bmatrix} a^2 \pm 2 a b - b^2 & \pm (a^2 + b^2) \\ \pm (a^2 + b^2) & -a^2 \pm 2 a b + b^2 \end{bmatrix}.
  \]
  Since $a^4 \ne b^4$, $\plholant{g}{\widehat{\EQ}}$ is $\SHARPP$-hard by Theorem~\ref{thm:Pl-Graph_Homomorphism},
  so $\PlHolant(\mathcal{F} \union \widehat{\EQ})$ is also $\SHARPP$-hard.
  
  \item Suppose $f = [1,0,-1,0,\dotsc]$ of arity at least~3.
  Connecting some number of $[1,0]$ gives $g = [1,0,-1,0]$ of arity exactly~3.
  Connecting $[a,b]$ to $g$ gives $h = [a,-b,-a]$.
  Since $a^4 \ne b^4$, $\plholant{h}{\widehat{\EQ}}$ is $\SHARPP$-hard by Theorem~\ref{thm:Pl-Graph_Homomorphism},
  so $\PlHolant(\mathcal{F} \union \widehat{\EQ})$ is also $\SHARPP$-hard.
  
  \item The argument for $f = [0,1,0,-1,\dotsc]$ is similar to the previous case.
  \qedhere
 \end{enumerate}
\end{proof}

\begin{lemma} \label{lem:mixing:unaries:MA}
 Suppose $f \in \mathscr{M} - \mathscr{A}$.
 If $a b \ne 0$,
 then for any set $\mathcal{F}$ of complex-weighted symmetric signatures containing $f$ and $[a,b]$,
 $\PlHolant(\mathcal{F} \union \widehat{\EQ})$ is $\SHARPP$-hard.
\end{lemma}

\begin{proof}
 Up to a nonzero scalar, the possibilities for $f$ are
 \begin{itemize}
  \item $[1, 0, r]$ with $r \ne 0$ and $r^4 \ne 1$;
  \item $[1, 0, r, 0, r^2, 0, \dotsc]$ of arity at least~3 with $r \ne 0$ and $r^2 \ne 1$;
  \item $[0, 1, 0, r, 0, r^2, \dotsc]$ of arity at least~3 with $r \ne 0$ and $r^2 \ne 1$;
  \item $[0, 1, 0, \dotsc, 0]$ of arity at least~3;
  \item $[0, \dotsc, 0, 1, 0]$ of arity at least~3.
 \end{itemize}
 We handle these cases below.
 
 \begin{enumerate}
  \item Suppose $f = [1,0,r]$ with $r^4 \ne 1$ and $r \ne 0$.
  Connecting $[a,b]$ to $[1,0,1,0]$ gives $[a,b,a]$ and connecting two copies of $[1,0,r]$ to $[a,b,a]$, one on each edge, gives $g = [a, b r, a r^2]$.
  If $a^2 \ne b^2$, then $\plholant{g}{\widehat{\EQ}}$ is $\SHARPP$-hard by Theorem~\ref{thm:Pl-Graph_Homomorphism},
  so $\PlHolant(\mathcal{F} \union \widehat{\EQ})$ is also $\SHARPP$-hard.
  
  Otherwise, $a^2 = b^2$ and we begin by connecting $[a,b]$ to $[1,0,r]$ to get $[a, b r]$.
  Then by the same construction, we have $g = [a, b r^2, a r^2]$ and
  $\plholant{g}{\widehat{\EQ}}$ is $\SHARPP$-hard by Theorem~\ref{thm:Pl-Graph_Homomorphism},
  so $\PlHolant(\mathcal{F} \union \widehat{\EQ})$ is also $\SHARPP$-hard.
  
  \item Suppose $f = [1,0,r,0,\dotsc]$ of arity at least~3 with $r^2 \ne 1$ and $r \ne 0$.
  Connecting some number of $[1,0]$ gives $g = [1,0,r,0]$ of arity exactly~3.
  Connecting $[a,b]$ to $g$ gives $h = [a, b r, a]$.
  If $a^2 \ne b^2 r$, then $\plholant{h}{\widehat{\EQ}}$ is $\SHARPP$-hard by Theorem~\ref{thm:Pl-Graph_Homomorphism},
  so $\PlHolant(\mathcal{F} \union \widehat{\EQ})$ is also $\SHARPP$-hard.
  
  Otherwise, $a^2 = b^2 r$ and we begin by connecting $[1,0]$ and $[a,b]$ to $[1,0,r,0]$ to get $[a, b r]$.
  Then by the same construction, we have $g = [a, b r^2, a r]$ and
  $\plholant{g}{\widehat{\EQ}}$ is $\SHARPP$-hard by Theorem~\ref{thm:Pl-Graph_Homomorphism},
  so $\PlHolant(\mathcal{F} \union \widehat{\EQ})$ is also $\SHARPP$-hard.
  
  \item The argument for $f = [0,1,0,r,\dotsc]$ is similar to the previous case.
  
  \item Suppose $f = [0,1,0,\dotsc,0]$ of arity $k \ge 3$.
  Connecting $k-2$ copies of $[a,b]$ to $f$ gives $g = a^{k-3}[(k-2) b,a,0]$.
  Since $a b \ne 0$, $\plholant{g}{\widehat{\EQ}}$ is $\SHARPP$-hard by Theorem~\ref{thm:Pl-Graph_Homomorphism},
  so $\PlHolant(\mathcal{F} \union \widehat{\EQ})$ is also $\SHARPP$-hard.
  
  \item The argument for $f = [0,\dotsc,0,1,0]$ is similar to the previous case.
  \qedhere
 \end{enumerate}
\end{proof}

Now we consider the general case of two signatures from two different tractable sets.
The three tractable sets give rise to three pairs of tractable sets to consider, each of which is covered in one of the next three lemmas.

\begin{lemma} \label{lem:mixing:AP}
 If $f \in \mathscr{A} - \widehat{\mathscr{P}}$ and $g \in \widehat{\mathscr{P}} - \mathscr{A}$,
 then for any set $\mathcal{F}$ of complex-weighted symmetric signatures containing $f$ and $g$,
 $\PlHolant(\mathcal{F} \union \widehat{\EQ})$ is $\SHARPP$-hard.
\end{lemma}

\begin{proof}
 The only possibility for $g$ is $[a,b,a,b,\dotsc]$, where $a b \ne 0$ and $a^4 \ne b^4$.
 Connecting some number of $[1,0]$ to $g$ gives $[a,b]$ and we are done by Lemma~\ref{lem:mixing:unaries:AP}.
\end{proof}

\begin{lemma} \label{lem:mixing:AM}
 If $f \in \mathscr{A} - \mathscr{M}$ and $g \in \mathscr{M} - \mathscr{A}$,
 then for any set $\mathcal{F}$ of complex-weighted symmetric signatures containing $f$ and $g$,
 $\PlHolant(\mathcal{F} \union \widehat{\EQ})$ is $\SHARPP$-hard.
\end{lemma}

\begin{proof}
 If $f$ does not contain a~0 entry, then after connecting some number of $[1,0]$ to $f$, we have a unary signature $[a,b]$ with $a b \ne 0$.
 Then we are done by Lemma~\ref{lem:mixing:unaries:MA}.
 
 Otherwise, $f$ contains a~0 entry.
 Then $f = [x, 0, \dotsc, 0, y]$ of arity at least~3 with $x y \ne 0$ (and $x^4 = y^4$).
 Up to a nonzero scalar, the possibilities for $g$ are
 \pagebreak
 \begin{itemize}
  \item $[1, 0, r]$ with $r \ne 0$ and $r^4 \ne 1$;
  \item $[1, 0, r, 0, r^2, 0, \dotsc]$ of arity at least~3 with $r \ne 0$ and $r^2 \ne 1$;
  \item $[0, 1, 0, r, 0, r^2, \dotsc]$ of arity at least~3 with $r \ne 0$ and $r^2 \ne 1$;
  \item $[0, 1, 0, 0, \dotsc, 0]$ of arity at least~3;
  \item $[0, \dotsc, 0, 0, 1, 0]$ of arity at least~3.
 \end{itemize}
 We handle these cases below.
 
 \begin{enumerate}
  \item Suppose $g = [1,0,r]$ with $r \ne 0$ and $r^4 \ne 1$.
  Then we are done by Corollary~\ref{cor:domain_pairing}.
  
  \item Suppose $g = [1,0,r,0,\dotsc]$ of arity at least~3 with $r \ne 0$ and $r^2 \ne 1$.
  After connecting some number of $[1,0]$ to $g$, we have $h = [1,0,r,0]$ of arity exactly~3.
  Then $\PlHolant(\{f,h\} \union \widehat{\EQ})$ is $\SHARPP$-hard by Lemma~\ref{lem:domain_pairing:strong},
  so $\PlHolant(\mathcal{F} \union \widehat{\EQ})$ is also $\SHARPP$-hard.
  
  \item Suppose $g = [0,1,0,r,\dotsc]$ of arity at least~3 with $r \ne 0$ and $r^2 \ne 1$.
  After connecting some number of $[1,0]$ to $g$, we have $h = [0,1,0,r]$ of arity exactly~3.
  Connecting two more copies of $[1,0]$ to $h$ gives $[0,1]$.
  Then we apply a holographic transformation by $T = \left[\begin{smallmatrix} 0 & 1 \\ 1 & 0 \end{smallmatrix}\right]$,
  so $f$ is transformed to $\hat{f} = [y, 0, \dotsc, 0, x]$ and $h$ is transformed to $\hat{h} = [r,0,1,0]$.
  Every even arity signature in $\widehat{\EQ}$ remains unchanged after a holographic transformation by $T$.
  By attaching $[0,1] T = [1,0]$ to every even arity signature in $T \widehat{\EQ}$, we obtain all of the odd arity signatures in $\widehat{\EQ}$ again.
  Then $\PlHolant(\{\hat{f}, \hat{h}\} \union \widehat{\EQ})$ is $\SHARPP$-hard by Lemma~\ref{lem:domain_pairing:strong},
  so $\PlHolant(\mathcal{F} \union \widehat{\EQ})$ is also $\SHARPP$-hard.
  
  \item Suppose $g = [0,1,0,\dotsc,0]$ of arity $k \ge 3$.
  The gadget in Figure~\ref{fig:gadget:double} with $g$ assigned to both vertices has signature $h = [k-1,0,1]$.
  Then $\PlHolant(\{f,h\} \union \widehat{\EQ})$ is $\SHARPP$-hard by Corollary~\ref{cor:domain_pairing},
  so $\PlHolant(\mathcal{F} \union \widehat{\EQ})$ is also $\SHARPP$-hard.
  
  \item The argument for $g = [0,\dotsc,0,1,0]$ is similar to the previous case.
  \qedhere
 \end{enumerate}
\end{proof}

\begin{figure}[t]
 \centering
 \begin{tikzpicture}[scale=\scale,transform shape,node distance=\nodeDist,semithick]
  \node[external] (0)              {};
  \node[internal] (1) [right of=0] {};
  \node[external] (2) [right of=1] {};
  \node[external] (3) [right of=2] {};
  \node[internal] (4) [right of=3] {};
  \node[external] (5) [right of=4] {};
  \path (0) edge                          node[near end]   (e1) {}               (1)
        (1) edge[out= 45, in= 135]        node             (e2) {}               (4)
            edge[out= 15, in= 165]                                               (4)
            edge[out=-10, in=-170, white] node[black]           {\Huge $\vdots$} (4)
            edge[out=-45, in=-135]        node             (e3) {}               (4)
        (4) edge                          node[near start] (e4) {}               (5);
  \begin{pgfonlayer}{background}
   \node[draw=\borderColor,thick,rounded corners,fit = (e1) (e2) (e3) (e4)] {};
  \end{pgfonlayer}
 \end{tikzpicture}
 \caption{The vertices are assigned $g = [0,1,0,\dotsc,0]$.}
 \label{fig:gadget:double}
\end{figure}
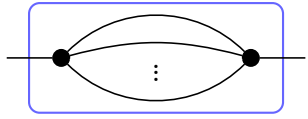

\begin{lemma} \label{lem:mixing:MP}
 Suppose $f \in \mathscr{M} - \widehat{\mathscr{P}}$ and $g \in \widehat{\mathscr{P}} - \mathscr{M}$ such that $\{f, g\} \not\subseteq \mathscr{A}$.
 Then for any set $\mathcal{F}$ of complex-weighted symmetric signatures containing $f$ and $g$,
 $\PlHolant(\mathcal{F} \union \widehat{\EQ})$ is $\SHARPP$-hard.
\end{lemma}

\begin{proof}
 The only possibility for $g$ is $[a,b,a,b,\dotsc]$, where $a b \ne 0$.
 Connecting some number of $[1,0]$ to $g$ gives $h = [a,b]$.
 If $f \not\in \mathscr{A}$, then $\PlHolant(\{f,h\} \union \widehat{\EQ})$ is $\SHARPP$-hard by Lemma~\ref{lem:mixing:unaries:MA},
 so $\PlHolant(\mathcal{F} \union \widehat{\EQ})$ is also $\SHARPP$-hard.
 
 Otherwise, $f \in \mathscr{A}$, so $g \not\in \mathscr{A}$.
 Then $\PlHolant(\{g,h\} \union \widehat{\EQ})$ is $\SHARPP$-hard by Lemma~\ref{lem:mixing:AP},
 so $\PlHolant(\mathcal{F} \union \widehat{\EQ})$ is also $\SHARPP$-hard.
\end{proof}

We summarize this section with the following theorem, which says that the tractable signature sets cannot mix.
More formally, signatures from different tractable sets, when put together, lead to $\SHARPP$-hardness.

\begin{theorem}[Mixing] \label{thm:mixing}
 Let $\mathcal{F}$ be any set of symmetric, complex-valued signatures in Boolean variables.
 If $\mathcal{F} \subseteq \mathscr{A} \union \widehat{\mathscr{P}} \union \mathscr{M}$,
 then $\PlHolant(\mathcal{F} \union \widehat{\EQ})$ is $\SHARPP$-hard
 unless $\mathcal{F} \subseteq \mathscr{A}$, $\mathcal{F} \subseteq \widehat{\mathscr{P}}$, or $\mathcal{F} \subseteq \mathscr{M}$,
 in which case $\PlHolant^c(\mathcal{F} \union \widehat{\EQ})$ is tractable.
\end{theorem}

\begin{proof}
 If $\mathcal{F}$ is a subset of $\mathscr{A}$, $\widehat{\mathscr{P}}$, or $\mathscr{M}$, then the tractability is given in Theorem~\ref{thm:PlCSP_tractable:Hadamard}.
 Otherwise $\mathcal{F}$ is not a subset of $\mathscr{A}$, $\widehat{\mathscr{P}}$, or $\mathscr{M}$.
 Then $\mathcal{F}$ contains a signature $g \in (\widehat{\mathscr{P}} \union \mathscr{M}) - \mathscr{A}$ since $\mathcal{F} \not\subseteq \mathscr{A}$.
 Suppose $\mathcal{F}$ contains a signature $f \in \mathscr{A} - \widehat{\mathscr{P}} - \mathscr{M}$.
 If $g \in \widehat{\mathscr{P}} - \mathscr{A}$, then $\PlHolant(\mathcal{F} \union \widehat{\EQ})$ is $\SHARPP$-hard by Lemma~\ref{lem:mixing:AP}.
 Otherwise, $g \in \mathscr{M} - \mathscr{A}$ and $\PlHolant(\mathcal{F} \union \widehat{\EQ})$ is $\SHARPP$-hard by Lemma~\ref{lem:mixing:AM}.
 
 Now assume that $\mathcal{F} \subseteq \widehat{\mathscr{P}} \union \mathscr{M}$.
 Since $(\widehat{\mathscr{P}} \intersect \mathscr{M}) - \mathscr{A}$ is empty (see Figure~\ref{fig:venn_diagram} in Appendix~\ref{sec:venn_diagram}),
 either $g \in \widehat{\mathscr{P}} - \mathscr{M} - \mathscr{A}$ or $g \in \mathscr{M} - \widehat{\mathscr{P}} - \mathscr{A}$
 because $\mathcal{F}$ is not a subset of either $\mathscr{M}$ or $\widehat{\mathscr{P}}$.
 If $g \in \widehat{\mathscr{P}} - \mathscr{M} - \mathscr{A}$,
 then there exists a signature $f \in \mathscr{M} - \widehat{\mathscr{P}}$ since $\mathcal{F} \not\subseteq \widehat{\mathscr{P}}$.
 In which case, $\PlHolant(\mathcal{F} \union \widehat{\EQ})$ is $\SHARPP$-hard by Lemma~\ref{lem:mixing:MP}.
 Otherwise, $g \in \mathscr{M} - \widehat{\mathscr{P}} - \mathscr{A}$
 and there exists a signature $f \in \widehat{\mathscr{P}} - \mathscr{M}$ since $\mathcal{F} \not\subseteq \mathscr{M}$.
 In which case, $\PlHolant(\mathcal{F} \union \widehat{\EQ})$ is $\SHARPP$-hard by Lemma~\ref{lem:mixing:MP}.
\end{proof}

\section{Pinning for Planar Graphs} \label{sec:pinning}

The idea of ``pinning'' is a common reduction technique between counting problems.
For the \#CSP framework, pinning fixes some variables to specific values of the domain by means of the constant functions~\cite{BD07, DGJ09, BDGJR09, HL12}.
In particular, for counting graph homomorphisms, pinning is used when the input graph is connected and the target graph is disconnected.
In this case, pinning a vertex of the input graph to a vertex of the target graph
forces all the vertices of the input graph to map to the same connected component of the target graph~\cite{DG00, BG05, GGJT10, Thu10, CCL13}.
For the Boolean domain, the constant~0 and constant~1 functions are the signatures $[1,0]$ and $[0,1]$ respectively.

From these works, the most relevant pinning lemma for the $\PlCSP$ framework is by Dyer, Goldberg, and Jerrum in~\cite{DGJ09},
where they show how to pin in the $\CSP$ framework.
However, the proof of this pinning lemma is highly nonplanar.
Cai, Lu, and Xia~\cite{CLX10} overcame this difficultly in the proof of their dichotomy theorem for the real-weighted $\PlCSP$ framework by
first undergoing a holographic transformation by the Hadamard matrix $H = \left[\begin{smallmatrix} 1 & 1 \\ 1 & -1 \end{smallmatrix}\right]$
and then pinning in this Hadamard basis.\footnote{The pinning in~\cite{CLX10}, which is accomplished in Section~IV,
is not summarized in a single statement but is implied by the combination of all the results in that section.}
We stress that this holographic transformation is necessary.
Indeed, if one were able to pin in the standard basis of the $\PlCSP$ framework,
then $\P = \SHARPP$ would follow since $\PlCSP(\widehat{\mathscr{M}})$ is tractable
but $\PlCSP(\widehat{\mathscr{M}} \union \{[1,0], [0,1]\})$ is $\SHARPP$-hard by our main dichotomy in Theorem~\ref{thm:PlCSP}
(or, more specifically, by Lemma~\ref{lem:mixing:unaries:MA}).

Since $\PlCSP(\mathcal{F})$ is Turing equivalent to $\PlHolant(\mathcal{F} \union \EQ)$,
the expression of $\PlCSP(\mathcal{F})$ in the Hadamard basis is $\PlHolant(H \mathcal{F} \union \widehat{\EQ})$.
Then we already have $[1,0] \in \widehat{\EQ}$, so pinning in the Hadamard basis of $\PlCSP(\mathcal{F})$ amounts to obtaining the missing signature $[0,1]$.

\subsection{The Road to Pinning}

We begin the road to pinning with a lemma that assumes the presence of $[0,0,1] = [0,1]^{\otimes 2}$, which is the tensor product of two copies of $[0,1]$.
In our pursuit to realize $[0,1]$, this may be as close as we can get, such as when every signature has even arity.
Another roadblock to realizing $[0,1]$ is when every signature has even parity.
Recall that a signature has even parity if its support is on entries of even Hamming weight.
By a simple parity argument, gadgets constructed with signatures of even parity can only realize signatures of even parity.
However, if every signature has even parity and $[0,0,1]$ is present, then we can already prove a dichotomy.

\begin{lemma} \label{lem:dichotomy:even_parity}
 Suppose $\mathcal{F}$ is a set of symmetric signatures with complex weights containing $[0,0,1]$.
 If every signature in $\mathcal{F}$ has even parity,
 then either $\PlHolant(\mathcal{F} \union \widehat{\EQ})$ is $\SHARPP$-hard
 or $\mathcal{F}$ is a subset of $\mathscr{A}$, $\widehat{\mathscr{P}}$, or $\mathscr{M}$,
 in which case $\PlHolant^c(\mathcal{F} \union \widehat{\EQ})$ is tractable.
\end{lemma}

\begin{proof}
 The tractability is given in Theorem~\ref{thm:PlCSP_tractable:Hadamard}.
 If every non-degenerate signature in $\mathcal{F}$ is of arity at most~3,
 then $\mathcal{F} \subseteq \mathscr{M}$ since all signatures 
 in $\mathcal{F}$ satisfy the (even) parity condition.
 
 Otherwise $\mathcal{F}$ contains some non-degenerate signature of arity at least~4.
 For every signature $f \in \mathcal{F}$ with $f = [f_0, f_1, \dotsc, f_m]$ and $m \ge 4$, using $[0,0,1]$ and $[1,0]$,
 we can obtain all subsignatures of the form $[f_{k-2},0,f_k,0,f_{k+2}]$ for any even $k$ such that $2 \le k \le m-2$.
 If any subsignature $g$ of this form satisfies $f_{k-2} f_{k+2} \ne f_k^2$ and $f_k \ne 0$,
 then $\PlHolant(g)$ is $\SHARPP$-hard by Corollary~\ref{cor:arity4:nonsingular_compressed_hard},
 so $\PlHolant(\mathcal{F} \union \widehat{\EQ})$ is also $\SHARPP$-hard.

 Otherwise all subsignatures of signatures in $\mathcal{F}$ of the above form satisfy $f_{k-2} f_{k+2} = f_k^2$ or $f_k = 0$.
 There are two types of signatures with this property.
 In the first type, the signature entries of even Hamming weight form a geometric progression.
 More specifically, the signatures of the first type have the form
 \[
  [\alpha^n, 0, \alpha^{n-1} \beta, 0, \dotsc, 0, \alpha \beta^{n-1}, 0, \beta^n]
  \qquad \text{or} \qquad
  [\alpha^n, 0, \alpha^{n-1} \beta, 0, \dotsc, 0, \alpha \beta^{n-1}, 0, \beta^n, 0]
 \]
 for some $\alpha, \beta \in \mathbb{C}$, which are in $\mathscr{M}$.
 In the second type, the signatures have arity at least~4 or~5 and 
 are of the form $[x,0,\dotsc,0,y]$ or $[x,0,\dotsc,0,y,0]$ respectively,
 with $x y \neq 0$ and an odd number of 0's between $x$ and $y$ (since they have even parity).
 If all of the signatures in $\mathcal{F}$ are of the first type,
 then $\mathcal{F} \subseteq \mathscr{M}$.
 
 Otherwise $\mathcal{F}$ contains a signature $f$ of the second type.
 Suppose $f = [x,0,\dotsc,0,y,0]$ of arity at least~$5$ with $x y \ne 0$.
 After some number of self-loops, we have $g = [x,0,0,0,y,0]$ of arity exactly~$5$.
 Then $\PlHolant(g)$ is $\SHARPP$-hard by Lemma~\ref{lem:arity5:hard_sig},
 so $\PlHolant(\mathcal{F} \union \widehat{\EQ})$ is also $\SHARPP$-hard.
 
 Otherwise $f = [x,0,\dotsc,0,y]$ of arity at least~$4$ with $x y \ne 0$.
 If $x^4 \ne y^4$, then $\PlHolant(\mathcal{F} \union \widehat{\EQ})$ is $\SHARPP$-hard by Lemma~\ref{lem:domain_pairing:weak}.
 
 Otherwise $x^4 = y^4$.
 This puts every signature of the second type in $\mathscr{A}$.
 Therefore $\mathcal{F} \subseteq \mathscr{A} \union \mathscr{M}$ and we are done by Theorem~\ref{thm:mixing}.
\end{proof}

The conclusion of every result in the rest of this section states that we are able to pin (under various assumptions on $\mathcal{F}$).
Formally speaking,
we repeatedly prove that $\PlHolant^c(\mathcal{F} \union \widehat{\EQ})$ is $\SHARPP$-hard (or in $\P$)
if and only if $\PlHolant(\mathcal{F} \union \widehat{\EQ})$ is $\SHARPP$-hard (or in $\P$).
The difference between these two counting problems is the presence of $[0,1]$ in $\PlHolant^c(\mathcal{F} \union \widehat{\EQ})$.
We always prove this statement in one of three ways:
\begin{enumerate}
 \item either we show that $\PlHolant^c(\mathcal{F} \union \widehat{\EQ})$ is tractable (so $\PlHolant(\mathcal{F} \union \widehat{\EQ})$ is as well);
 \item or we show that $\PlHolant(\mathcal{F} \union \widehat{\EQ})$ is $\SHARPP$-hard (so $\PlHolant^c(\mathcal{F} \union \widehat{\EQ})$ is as well);
 \item or we show how to reduce $\PlHolant^c(\mathcal{F} \union \widehat{\EQ})$ to $\PlHolant(\mathcal{F} \union \widehat{\EQ})$ by realizing $[0,1]$ using signatures in $\mathcal{F} \union \widehat{\EQ}$.
\end{enumerate}

\begin{lemma} \label{lem:pinning:001}
 Let $\mathcal{F}$ be any set of complex-weighted symmetric signatures containing $[0,0,1]$.
 Then $\PlHolant^c(\mathcal{F} \union \widehat{\EQ})$ is $\SHARPP$-hard (or in $\P$)
 iff $\PlHolant(\mathcal{F} \union \widehat{\EQ})$ is $\SHARPP$-hard (or in $\P$).
\end{lemma}

\begin{proof}
 If we had a unary signature $[a,b]$ where $b \ne 0$, then connecting $[a,b]$ to $[0,0,1]$ gives the signature $[0,b]$,
 which is $[0,1]$ after normalizing.
 Thus, in order to reduce $\PlHolant^c(\mathcal{F} \union \widehat{\EQ})$ to $\PlHolant(\mathcal{F} \union \widehat{\EQ})$ by constructing $[0,1]$,
 it suffices to construct a unary signature $[a,b]$ with $b \ne 0$.
 
 For every signature $f \in \mathcal{F}$ with $f = [f_0, f_1, \dotsc, f_m]$, using $[0,0,1]$ and $[1,0]$,
 we can obtain all subsignatures of the form $[f_{k-1},f_k]$ for any odd $k$ such that $1 \le k \le m$.
 If any subsignature satisfies $f_{k} \ne 0$, then we can construct $[0,1]$.
 
 Otherwise all signatures in $\mathcal{F}$ have even parity and we are done by Lemma~\ref{lem:dichotomy:even_parity}.
\end{proof}

\begin{figure}[t]
 \centering
 \begin{tikzpicture}[scale=\scale,transform shape,node distance=\nodeDist,semithick]
  \node[external]  (0)                    {};
  \node[internal]  (1) [right       of=0] {};
  \node[triangle]  (2) [above right of=1] {};
  \node[triangle]  (3) [below right of=1] {};
  \node[internal]  (4) [above right of=3] {};
  \node[external]  (5) [right       of=4] {};
  \path (0) edge node[near end]   (e1) {} (1)
        (1) edge                          (2)
            edge                          (3)
        (4) edge                          (2)
            edge                          (3)
            edge node[near start] (e2) {} (5);
  \begin{pgfonlayer}{background}
   \node[draw=\borderColor,thick,rounded corners,fit = (2) (3) (e1) (e2)] {};
  \end{pgfonlayer}
 \end{tikzpicture}
 \caption{The circles are assigned $[1,0,1,0]$ and the triangles are assigned $[1,0,x]$.}
 \label{fig:gadget:binary_interpolation}
\end{figure}
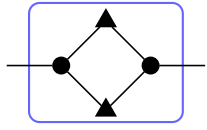

There are two scenarios that lead to Lemma~\ref{lem:pinning:001},
which are the focus of the next two lemmas.

\begin{lemma} \label{lem:pinning:binary_interpolation}
 For $x \in \mathbb{C}$, let $\mathcal{F}$ be any set of complex-weighted symmetric signatures containing $[1,0,x]$ such that $x \not\in \{0, \pm 1\}$.
 Then $\PlHolant^c(\mathcal{F} \union \widehat{\EQ})$ is $\SHARPP$-hard (or in $\P$)
 iff $\PlHolant(\mathcal{F} \union \widehat{\EQ})$ is $\SHARPP$-hard (or in $\P$).
\end{lemma}

\begin{proof}
 There are two cases.
 In either case, we realize $[0,0,1]$ and finish by applying Lemma~\ref{lem:pinning:001}.

 First we claim that the conclusion holds provided $|x| \ne 0,1$.
 Combining $k$ copies of $[1,0,x]$ gives $[1,0,x^k]$.
 Since $|x| \not\in \{0,1\}$, $x$ is neither zero nor a root of unity,
 so we can use polynomial interpolation to realize $[a,0,b]$ for any $a, b \in \mathbb{C}$, including $[0,0,1]$.
 
 Otherwise $|x| = 1$.
 The gadget in Figure~\ref{fig:gadget:binary_interpolation} has signature $[f_0, f_1, f_2] = [1 + x^2, 0, 2x]$.
 If $x = \pm i$, then we have $[0, 0, \pm 2 i]$, which is $[0,0,1]$ after normalizing.
 
 Otherwise $x \neq \pm i$, so $f_0 \neq 0$.
 Since $x \neq 0$, $f_2 \neq 0$.
 Since $x \neq \pm 1$, $|f_0| < 2$.
 However, $|f_2| = 2$.
 Therefore, after normalizing, the signature $[1,0,y]$ with $y = \frac{2x}{1 + x^2}$ has $|y| > 1$,
 so it can interpolate $[0,0,1]$ by our initial claim since $|y| \not\in \{0,1\}$.
\end{proof}

\begin{lemma} \label{lem:pinning:f0=0}
 Let $\mathcal{F}$ be any set of complex-weighted symmetric signatures containing a signature $[f_0, f_1, \dotsc, f_n]$
 that is not identically~0 but has $f_0 = 0$.
 Then $\PlHolant^c(\mathcal{F} \union \widehat{\EQ})$ is $\SHARPP$-hard (or in $\P$)
 iff $\PlHolant(\mathcal{F} \union \widehat{\EQ})$ is $\SHARPP$-hard (or in $\P$).
\end{lemma}

\begin{proof}
 If $f_1 \ne 0$, then we connect $n-1$ copies of $[1,0]$ to $f$ to get $[0,f_1]$, which is $[0,1]$ after normalizing.
 If $f_1 = 0$, then $n \ge 2$.
 If $f_2 \ne 0$, then we connect $n-2$ copies of $[1,0]$ to $f$ to get $[0,0,f_2]$, which is $[0,0,1]$ after normalizing.
 Then we are done by Lemma~\ref{lem:pinning:001}.
 If $f_1 = f_2 = 0$, then $n \ge 3$.
 With some number of self-loops, we get a signature with exactly one or two initial zeros, which is one of the above scenarios.
\end{proof}

As a significant step toward pinning for any signature set $\mathcal{F}$, we show how to pin given any binary signature.
Some cases resist pinning and are excluded.

\begin{lemma} \label{lem:pinning:binary}
 Let $\mathcal{F}$ be any set of complex-weighted symmetric signatures containing $f = [f_0, f_1, f_2]$.
 Then $\PlHolant^c(\mathcal{F} \union \widehat{\EQ})$ is $\SHARPP$-hard (or in $\P$)
 iff $\PlHolant(\mathcal{F} \union \widehat{\EQ})$ is $\SHARPP$-hard (or in $\P$)
 unless $f \in \{[0,0,0], [1,0,-1], [1,r,r^2], [1,b,1]\}$, up to a nonzero scalar, for any $b, r \in \mathbb{C}$.
\end{lemma}

\begin{proof}
 If $f_0 = 0$ and either $f_1 \ne 0$ or $f_2 \ne 0$,
 then we are done by Lemma~\ref{lem:pinning:f0=0}.
 Otherwise, $f = [0,0,0]$ or $f_0 \ne 0$,
 in which case we normalize $f_0$ to~$1$.
 If $\plholant{f}{\widehat{\EQ}}$ is $\SHARPP$-hard by Theorem~\ref{thm:Pl-Graph_Homomorphism},
 then $\PlHolant(\mathcal{F} \union \widehat{\EQ})$ is also $\SHARPP$-hard.
 Otherwise, $f$ is one of the tractable cases, which implies that
 \[f \in \{[0,0,0], [1,r,r^2], [1,0,x], [1, \pm 1, -1], [1,b,1]\}.\]
 If $f = [1, \pm 1, -1]$, then we connect $f$ to $[1,0,1,0]$ to get $[0, \pm 2]$, which is $[0,1]$ after normalizing.
 If $f = [1, 0, x]$, then we are done by Lemma~\ref{lem:pinning:binary_interpolation} unless $x \in \{0, \pm 1\}$.
 The remaining cases are all excluded by assumption, so we are done.
\end{proof}

\subsection{Pinning in the Hadamard Basis}

Before we show how to pin in the Hadamard basis, we handle two simple cases.

\begin{lemma} \label{lem:pinning:1i}
 For any set $\mathcal{F}$ of complex-weighted symmetric signatures containing $[1, \pm i]$, we have
 $\PlHolant^c(\mathcal{F} \union \widehat{\EQ}) \le_T \PlHolant(\mathcal{F} \union \widehat{\EQ})$.
\end{lemma}

\begin{proof}
 Connect two copies of $[1, \pm i]$ to $[1,0,1,0]$ to get $[0, \pm 2 i]$, which is $[0,1]$ after normalizing.
\end{proof}

The next lemma considers the signature $[1,b,1,b^{-1}]$,
which we also encounter in Theorem~\ref{thm:dichotomy:single_sig}, the single signature dichotomy.

\begin{lemma} \label{lem:pinning:1b1bInverse}
 Let $b \in \mathbb{C}$.
 If $b \not\in \{0, \pm 1\}$,
 then for any set $\mathcal{F}$ of complex-weighted symmetric signatures containing $f = [1,b,1,b^{-1}]$,
 $\PlHolant(\mathcal{F} \union \widehat{\EQ})$ is $\SHARPP$-hard.
\end{lemma}

\begin{proof}
 Connect two copies of $[1,0]$ to $f$ to get $[1, b]$.
 Connecting this back to $f$ gives $g = [1 + b^2, 2 b, 2]$.
 Then $\plholant{g}{\widehat{\EQ}}$ is $\SHARPP$-hard by Theorem~\ref{thm:Pl-Graph_Homomorphism},
 so $\PlHolant(\mathcal{F} \union \widehat{\EQ})$ is also $\SHARPP$-hard.
\end{proof}

Now we are ready to prove our pinning result.

\begin{theorem}[Pinning] \label{thm:pinning}
 Let $\mathcal{F}$ be any set of complex-weighted symmetric signatures.
 Then $\PlHolant^c(\mathcal{F} \union \widehat{\EQ})$ is $\SHARPP$-hard (or in $\P$)
 iff $\PlHolant(\mathcal{F} \union \widehat{\EQ})$ is $\SHARPP$-hard (or in $\P$).
\end{theorem}

This theorem does not exclude the possibility that either framework can express a problem of intermediate complexity.
It merely says that if one framework does not contain a problem of intermediate complexity,
then neither does the other.
Our goal is to prove a dichotomy for $\PlHolant(\mathcal{F} \union \widehat{\EQ})$.
By Theorem~\ref{thm:pinning}, this is equivalent to proving a dichotomy for $\PlHolant^c(\mathcal{F} \union \widehat{\EQ})$.

\begin{proof}[Proof of Theorem~\ref{thm:pinning}.]
 For simplicity, we normalize the first nonzero entry of every signature in $\mathcal{F}$ to~$1$.
 If $\mathcal{F}$ contains the degenerate signature $[0,1]^{\otimes n}$ for some $n \ge 1$,
 then we take self-loops on this signature until we have either $[0,1]$ or $[0,0,1]$ (depending on the parity of $n$).
 If we have $[0,1]$, we are done.
 Otherwise, we have $[0,0,1]$ and are done by Lemma~\ref{lem:pinning:001}.
 
 Now assume that any degenerate signature in $\mathcal{F}$ is not of the form $[0,1]^{\otimes n}$.
 Then we can replace these degenerate signatures in $\mathcal{F}$ by their unary versions using $[1,0]$.
 This does not change the complexity of the problem.
 If $\mathcal{F}$ contains only unary signatures,
 then $\mathcal{F} \subseteq \widehat{\mathscr{P}}$ and $\PlHolant^c(\mathcal{F} \union \widehat{\EQ})$ is tractable by Theorem~\ref{thm:PlCSP_tractable:Hadamard}.
 
 Otherwise $\mathcal{F}$ contains a signature $f$ of arity at least two.
 We connect some number of $[1,0]$ to $f$ until we obtain a signature with arity exactly two.
 We call the resulting signature the binary prefix of $f$.
 If this binary prefix is not one of the exceptional forms in Lemma~\ref{lem:pinning:binary}, then we are done, so assume that it is one of the exceptional forms.
 
 Now we perform case analysis according to the exceptional forms in Lemma~\ref{lem:pinning:binary}.
 There are five cases below because we consider $[1,r,r^2]$ as $[1,0,0]$ and $[1,r,r^2]$ with $r \neq 0$ as separate cases.
 In each case, we either show that the conclusion of the theorem holds or that $f \in \mathscr{A} \union \widehat{\mathscr{P}} \union \mathscr{M}$.
 After the case analysis, we then handle all of these tractable $f$ together.
 
 \begin{enumerate}
  \item Suppose the binary prefix of $f$ is $[0,0,0]$.
  If $f$ is not identically~0, then we are done by Lemma~\ref{lem:pinning:f0=0}.
  
  Thus, in this case, we may assume $f = [0,0,\dotsc,0]$ is identically~0.
  
  \item Suppose the binary prefix of $f$ is $[1,0,-1]$.
  If $f$ is not of the form
  \begin{align}
   [1,0,-1,0,1,0,-1,0,\dotsc, 0 \text{ or } 1 \text{ or } (-1)], \label{equ:10-1_pattern}
  \end{align}
  then after one self-loop, we have a signature of arity at least one with~0 as its first entry but is not identically~0,
  so we are done by Lemma~\ref{lem:pinning:f0=0}.
  
  Thus, in this case, we may assume $f$ has the form given in~(\ref{equ:10-1_pattern}).
  
  \item Suppose the binary prefix of $f$ is $[1,0,0]$.
  If $f$ is not of the form $[1,0,\dotsc,0]$,
  then after connecting some number of $[1,0]$, we have $[1,0,\dotsc,0,x]$ of arity at least~3, where $x \ne 0$.
  If $x^4 \ne 1$, then $\PlHolant(\{f\} \union \widehat{\EQ})$ is $\SHARPP$-hard by Lemma~\ref{lem:domain_pairing:weak},
  so $\PlHolant(\mathcal{F} \union \widehat{\EQ})$ is also $\SHARPP$-hard.
  
  Otherwise, $x^4 = 1$.
  Suppose that $x$ is not the last entry in $f$.
  Then connecting one fewer $[1,0]$ than before, we have $g = [1,0,\dotsc,0,x,y]$ and there are two cases to consider.
  If the index of $x$ in $g$ is odd, then after some number of self-loops, we have $h = [1,0,0,x,y]$.
  The determinant of the compressed signature matrix of $h$ is $-2 x^2 \ne 0$.
  Thus, $\Holant(h)$ is $\SHARPP$-hard by Corollary~\ref{cor:arity4:nonsingular_compressed_hard},
  so $\PlHolant(\mathcal{F} \union \widehat{\EQ})$ is also $\SHARPP$-hard.
  
  Otherwise, the index of $x$ in $g$ is even.
  After some number of self-loops, we have $h = [1,0,0,0,x,y]$.
  Then by Lemma~\ref{lem:arity5:hard_sig}, $\Holant(h)$ is $\SHARPP$-hard,
  so $\PlHolant(\mathcal{F} \union \widehat{\EQ})$ is also $\SHARPP$-hard.
  
  Thus, in this case, we may assume either $f = [1,0,\dotsc,0]$ or $f = [1, 0, \dotsc, 0, x]$ with $x^4 = 1$.
  
  \item Suppose the binary prefix of $f$ is $[1,r,r^2]$, where $r \ne 0$.
  If $f$ is not of the form $[1,r,\dotsc,r^n]$,
  then after connecting some number of $[1,0]$, we have $[1,r,\dotsc,r^m,y]$, where $y \ne r^{m+1}$ and $m \ge 2$.
  Using $[1,0]$, we can get $[1,r]$.
  If $r = \pm i$, then we are done by Lemma~\ref{lem:pinning:1i}, so assume that $r \ne \pm i$.
  Then we can attach $[1,r]$ back to the initial signature some number of times to get $g = [1,r,r^2,x]$ after normalizing, where $x \ne r^3$.
  We connect $[1,r]$ once more to get $h = [1 + r^2, r (1 + r^2), r^2 + r x]$.
  If $h$ does not have one of the exceptional forms in Lemma~\ref{lem:pinning:binary}, then we are done, so assume that it does.
  
  Since the second entry of $h$ is not~0 and $x \ne r^3$, the only possibility is that $h$ has the form $[1,b,1]$ up to a scalar.
  This gives $x = r^{-1}$.
  Note that $r \ne \pm 1$ since $x \ne r^3$.
  A self-loop on $g = [1,r,r^2,r^{-1}]$ gives $[1 + r^2, r + r^{-1}]$, which is $[1,r^{-1}]$ after normalizing.
  Connecting this back to $g$ gives $h = [2, 2 r, r^2 + r^{-2}]$.
  We assume that $h$ has one of the exceptional forms in Lemma~\ref{lem:pinning:binary} since we are done otherwise.
  If $h$ has the form $[1,r,r^2]$ up to a scalar, then $r^4 = 1$, a contradiction, so it must have the form $[1,b,1]$ up to a scalar.
  But then $r^2 = 1$, which is also a contradiction.
  
  Thus, in this case, we may assume $f = [1,r,\dotsc,r^n]$.
  
  \item Suppose the binary prefix of $f$ is $[1,b,1]$.
  If $b = \pm 1$, then this binary prefix is degenerate and was considered in the previous case, so assume that $b \ne \pm 1$.
  If $f$ is not of the form $[1,b,1,b,\dotsc]$,
  then suppose that the index of the first entry in $f$ to break the pattern is even.
  Then after connecting some number of $[1,0]$, we have $[1,b,1,\dotsc,b,y]$, where $y \ne 1$.
  Then after some number of self-loops and normalizing, we have $g = [1,b,1,b,x]$, where $x \ne 1$.
  The determinant of its compressed signature matrix is $(b^2 - 1) (1 - x) \ne 0$.
  Thus, $\Holant(g)$ is $\SHARPP$-hard by Corollary~\ref{cor:arity4:nonsingular_compressed_hard},
  so $\PlHolant(\mathcal{F} \union \widehat{\EQ})$ is also $\SHARPP$-hard.
  
  Otherwise, the index of the first entry in $f$ to break the pattern is odd.
  Then after connecting some number of $[1,0]$, we have $[1,b,1,\dotsc,1,y]$, where $y \ne b$.
  Then after some number of self-loops and normalizing, we have $[1,b,1,x]$, where $x \ne b$.
  We do a self-loop to get $g = [2, b + x]$.
  If $b = 0$, then connecting $g$ to $[1,0,1,x]$ gives $h = [2, x, 2 + x^2]$.
  We assume that $h$ has one of the exceptional forms in Lemma~\ref{lem:pinning:binary} since we are done otherwise.
  Because $x \ne 0$, the only possibility is that $h$ has the form $[1,r,r^2]$ up to a scalar.
  Then we get $x^2 = -4$, so $g = [2,x] = 2 [1, \pm i]$ and we are done by Lemma~\ref{lem:pinning:1i}.
  We use the signature $g$ again below.
 
  Otherwise, $b \ne 0$.
  Using $[1,0]$, we can get $h = [1,b,1]$.
  If the signature matrix $M_h$ of $h$ has finite order modulo a scalar,
  then $M_h^\ell = \beta I_2$ for some positive integer $\ell$ and some nonzero complex value $\beta$.
  Thus after normalizing, we can construct the anti-gadget $[1,-b,1]$ by connecting $\ell - 1$ copies of $h$ together.
  Connecting $[1,0]$ to $[1,-b,1]$ gives $[1,-b]$ and connecting this to $[1,b,1,x]$ gives $[1 - b^2, 0, 1 - b x]$.
  If $\frac{1 - b x}{1 - b^2} \not\in \{0, \pm 1\}$, then we are done by Lemma~\ref{lem:pinning:binary_interpolation}.
  
  Otherwise, $y = \frac{1 - b x}{1 - b^2} \in \{0, \pm 1\}$.
  For $y = 0$, we get $x = b^{-1}$ and are done by Lemma~\ref{lem:pinning:1b1bInverse} since $b \not\in \{0, \pm 1\}$.
  For $y = 1$, we get $b = x$, a contradiction.
  For $y = -1$, we get $2 - b^2 - b x = 0$.
  Then connecting $[1,-b,1]$ to $g = [2, b + x]$ gives $[2 - b^2 - b x, x - b] = [0, x - b]$, which is $[0,1]$ after normalizing.
  
  Otherwise, $M_h$ has infinite order modulo a scalar.
  Then we can interpolate $[0,1]$ by Lemma~\ref{lem:unary_interpolation:EQ-hat} since $b \not\in \{0, \pm 1\}$.
  
  Thus, in this case, we may assume $f = [1,b,1,b,\dotsc]$.
 \end{enumerate}
 
 At this point, every signature in $\mathcal{F}$ (including the unary signatures) must be of one of the following forms:
 \begin{itemize}
  \item $[0, \dotsc, 0]$, which is in $\mathscr{A} \intersect \widehat{\mathscr{P}} \intersect \mathscr{M}$;
  \item $[1, 0, -1, 0, 1, 0, -1, 0, \dotsc, 0 \text{ or } 1 \text{ or } (-1)]$, which is in $\mathscr{A} \intersect \mathscr{M}$;
  \item $[1, 0, \dotsc, 0, x]$, where $x^4 = 1$, which is in $\mathscr{A}$;
  \item $[1, b, 1, b, \dotsc, 1 \text{ or } b]$, which is in $\widehat{\mathscr{P}}$.
 \end{itemize}
 In particular, every possible unary signature either fits into the first case or the last case.
 Therefore $\mathcal{F} \subseteq \mathscr{A} \union \widehat{\mathscr{P}} \union \mathscr{M}$ and we are done by Theorem~\ref{thm:mixing}.
\end{proof}

\section{Main Dichotomy} \label{sec:dichotomy}

In this section, we prove our main dichotomy theorem.
We begin with a dichotomy for a single signature.

\begin{theorem} \label{thm:dichotomy:single_sig}
 If $f$ is a non-degenerate symmetric signature of arity at least~2 with complex weights in Boolean variables,
 then $\PlHolant(\{f\} \union \widehat{\EQ})$ is $\SHARPP$-hard unless $f \in \mathscr{A} \union \widehat{\mathscr{P}} \union \mathscr{M}$,
 in which case the problem is in $\P$.
\end{theorem}

\begin{proof}
 When $f \in \mathscr{A} \union \widehat{\mathscr{P}} \union \mathscr{M}$, the problem is tractable by Theorem~\ref{thm:PlCSP_tractable:Hadamard}.
 When $f \not\in \mathscr{A} \union \widehat{\mathscr{P}} \union \mathscr{M}$,
 we prove that $\PlHolant^c(\{f\} \union \widehat{\EQ})$ is $\SHARPP$-hard,
 which is sufficient because of pinning (Theorem~\ref{thm:pinning}).
 Using $[1,0]$ and $[0,1]$, we can obtain any subsignature of $f$.
 
 Notice that once we have $[0,1]$ and $\widehat{\EQ}$, we can realize every signature in $T \widehat{\EQ}$,
 where $T = \left[\begin{smallmatrix} 0 & 1 \\ 1 & 0 \end{smallmatrix}\right]$.
 In fact, every even arity signature in $\widehat{\EQ}$ is also in $T \widehat{\EQ}$, and 
 we obtain all the odd arity signatures in $T \widehat{\EQ}$ by attaching $[0,1]$ to all the even arity signatures in $\widehat{\EQ}$.
 Therefore, a holographic transformation by $T$ does not change the complexity of the problem.
 Furthermore, $\mathscr{A} \union \widehat{\mathscr{P}} \union \mathscr{M}$ is closed under $T$.
 We use these facts later.
 
 The possibilities for $f$ can be divided into three cases:
 \begin{itemize}
  \item $f$ satisfies the parity condition;
  \item $f$ does not satisfy the parity condition but does contain a~$0$ entry;
  \item $f$ does not contain a~$0$ entry.
 \end{itemize}
 We handle these cases below.

 \begin{enumerate}
  \item Suppose that $f$ satisfies the parity condition.
  If $f$ has even parity, then we are done by Lemma~\ref{lem:dichotomy:even_parity}.
 
  Otherwise, $f$ has odd parity.
  If $f$ has odd arity, then under a holographic transformation by $T = \left[\begin{smallmatrix} 0 & 1 \\ 1 & 0 \end{smallmatrix}\right]$,
  $f$ is transformed to $\hat{f}$, which has even parity.
  Then either $\PlHolant^c(\{\hat{f}\} \union \widehat{\EQ})$ is $\SHARPP$-hard by Lemma~\ref{lem:dichotomy:even_parity}
  (and thus $\PlHolant^c(\{f\} \union \widehat{\EQ})$ is also $\SHARPP$-hard),
  or $\hat{f} \in \mathscr{A} \union \widehat{\mathscr{P}} \union \mathscr{M}$ (and thus $f \in \mathscr{A} \union \widehat{\mathscr{P}} \union \mathscr{M}$).
 
  Otherwise, the arity of $f$ is even.
  Connect $[0,1]$ to $f$ to get a signature $g$ with even parity and odd arity.
  Then either $\PlHolant^c(\{g\} \union \widehat{\EQ})$ is $\SHARPP$-hard by Lemma~\ref{lem:dichotomy:even_parity}
  (and thus $\PlHolant^c(\{f\} \union \widehat{\EQ})$ is also $\SHARPP$-hard),
  or $g \in \mathscr{A} \union \widehat{\mathscr{P}} \union \mathscr{M}$.
  In the latter case,
  it must be that $g \in \mathscr{M}$ since non-degenerate generalized equality signatures cannot have both even parity and odd arity.
  (See Figure~\ref{fig:venn_diagram} at the end of the Appendix,
  which contains a Venn diagram of the signatures in $\mathscr{A} \union \widehat{\mathscr{P}} \union \mathscr{M}$, up to constant factors.)
  In particular, the even parity entries of $g$ form a geometric progression.
  Therefore $f \in \mathscr{M}$ since $f$ has odd parity and the same geometric progression among its odd parity entries.
  
  \item Suppose that $f$ contains a~0 entry but does not satisfy the parity condition.
  Since $f$ does not satisfy the parity condition, there must be at least two nonzero entries separated by an even number of~$0$ entries.
  Thus, $f$ contains a subsignature $g = [a,0,\dotsc,0,b]$ of arity $n = 2 k + 1 \ge 1$, where $a b \ne 0$.
  
  If $k = 0$, then $n = 1$ and we can shift either to the right or to the left
  and find the~0 entry in $f$ and obtain a binary subsignature $h$ of the form $[c,d,0]$ or $[0,c,d]$,
  where $c d \ne 0$.
  Then $\plholant{h}{\widehat{\EQ}}$ is $\SHARPP$-hard by Theorem~\ref{thm:Pl-Graph_Homomorphism},
  so $\PlHolant(\{f\} \union \widehat{\EQ})$ is also $\SHARPP$-hard.
  
  Otherwise $k \ge 1$, so $n \ge 3$.
  If $a^4 \ne b^4$, then $\PlHolant(\{g\} \union \widehat{\EQ})$ is $\SHARPP$-hard by Lemma~\ref{lem:domain_pairing:weak},
  so $\PlHolant(\{f\} \union \widehat{\EQ})$ is also $\SHARPP$-hard.
  
  Otherwise, $a^4 = b^4$, so $g \in \mathscr{A}$.
  If $f = g$, then we are done, so assume that $f \ne g$, which implies that there is another entry just before $a$ or just after $b$.
  If this entry is nonzero, then $f$ has a subsignature $h$ of the form $[c,a,0]$ or $[0,b,d]$, where $c d \ne 0$.
  Then $\plholant{h}{\widehat{\EQ}}$ is $\SHARPP$-hard by Theorem~\ref{thm:Pl-Graph_Homomorphism},
  so $\PlHolant(\{f\} \union \widehat{\EQ})$ is also $\SHARPP$-hard.
  
  Otherwise, this entry is~$0$ and $f$ has a subsignature $h$ of the form $[0,a,0,\dotsc,0,b]$ or $[a,0,\dotsc,0,b,0]$ of arity at least~4.
  If the arity of $h$ is even, then after some number of self-loops,
  we have a signature $h'$ of the form $[0,a,0,0,b]$ or $[a,0,0,b,0]$ of arity exactly~4.
  Then $\PlHolant(h')$ is $\SHARPP$-hard by Corollary~\ref{cor:arity4:nonsingular_compressed_hard} since $a b \ne 0$,
  so $\PlHolant(\{f\} \union \widehat{\EQ})$ is also $\SHARPP$-hard.
  
  Otherwise, the arity of $h$ is odd.
  After some number of self-loops,
  we have a signature $h'$ of the form $[0,a,0,0,0,b]$ or $[a,0,0,0,b,0]$ of arity exactly~5.
  Then $\PlHolant(h')$ is $\SHARPP$-hard by Lemma~\ref{lem:arity5:hard_sig} since $a b \ne 0$,
  so $\PlHolant(\{f\} \union \widehat{\EQ})$ is also $\SHARPP$-hard.
  
  \item Suppose $f$ contains no~0 entry.
  If $f$ has a binary subsignature $g$ such that $\plholant{g}{\widehat{\EQ}}$ is $\SHARPP$-hard by Theorem~\ref{thm:Pl-Graph_Homomorphism},
  then $\PlHolant(\{f\} \union \widehat{\EQ})$ is also $\SHARPP$-hard.
  
  Otherwise every binary subsignature $[a,b,c]$ of $f$ satisfies the conditions of some tractable case in Theorem~\ref{thm:Pl-Graph_Homomorphism}.
  The three possible tractable cases are
  degenerate with condition $a c = b^2$ (case~\ref{case:Pl-GH:degenerate}),
  affine $\mathscr{A}$ with condition $a c = -b^2 \wedge a = -c$ (case~\ref{case:Pl-GH:A}),
  and a Hadamard-transformed product type $\widehat{\mathscr{P}}$ with condition $a = c$ (case~\ref{case:Pl-GH:P-hat}).
  If every binary subsignature $[a,b,c]$ of $f$ satisfies $a c = b^2$,
  then $f$ is degenerate, a contradiction.
  If every binary subsignature $[a,b,c]$ of $f$ satisfies $a c = -b^2 \wedge a = -c$,
  then $f = [1, \pm 1, -1, \mp 1, 1, \pm 1, -1, \mp 1, \dotsc] \in \mathscr{A}$ (up to a scalar) and we are done.
  If every binary subsignature $[a,b,c]$ of $f$ satisfies $a = c$,
  then $f \in \widehat{\mathscr{P}}$ and we are done.
  
  Otherwise, there exists two binary subsignatures of $f$ that do not satisfy the same tractable case in Theorem~\ref{thm:Pl-Graph_Homomorphism}.
  More specifically, $f$ has arity at least~3 and there exists a ternary subsignature $g = [a,b,c,d]$
  such that $h = [a,b,c]$ and $h' = [b,c,d]$ exclusively satisfy the conditions of different tractable cases in Theorem~\ref{thm:Pl-Graph_Homomorphism}.
  By symmetry in the statement of the tractable conditions in Theorem~\ref{thm:Pl-Graph_Homomorphism},
  under a holographic transformation by $\left[\begin{smallmatrix} 0 & 1 \\ 1 & 0 \end{smallmatrix}\right]$,
  we can we can switch the order of $h$ and $h'$.
  Suppose $f$ contains a binary subsignature that satisfies the condition of the affine case.
  Let $h$ be that subsignature.
  Then for either case of $h'$, we have $g = [1, \varepsilon, -1, \varepsilon]$ after normalizing, where $\varepsilon^2 = 1$.
  Connecting two copies of $[0,1]$ to $g$ gives $[-1, \varepsilon]$.
  Connecting this back to $g$ gives $g' = [0, -2 \varepsilon, 2]$.
  Then $\plholant{g'}{\widehat{\EQ}}$ is $\SHARPP$-hard by Theorem~\ref{thm:Pl-Graph_Homomorphism},
  so $\PlHolant(\{f\} \union \widehat{\EQ})$ is also $\SHARPP$-hard.
 
  Otherwise, we may assume that $h$ satisfies the product type condition (but not the degenerate condition) and $h'$ satisfies the degenerate condition.
  Then $g = [1, b, 1, b^{-1}]$ after normalizing, where $b^2 \ne 1$.
  Then $\plholant{g}{\widehat{\EQ}}$ is $\SHARPP$-hard by Lemma~\ref{lem:pinning:1b1bInverse},
  so $\PlHolant(\{f\} \union \widehat{\EQ})$ is also $\SHARPP$-hard.
  \qedhere
 \end{enumerate}
\end{proof}

Now we are ready to prove our main dichotomy theorem.

\begin{theorem} \label{thm:PlCSP:Hadamard}
 Let $\mathcal{F}$ be any set of symmetric, complex-valued signatures in Boolean variables.
 Then $\PlHolant(\mathcal{F} \union \widehat{\EQ})$ is $\SHARPP$-hard unless
 $\mathcal{F} \subseteq \mathscr{A}$, $\mathcal{F} \subseteq \widehat{\mathscr{P}}$, or $\mathcal{F} \subseteq \mathscr{M}$,
 in which case the problem is in $\P$.
\end{theorem}

\begin{proof}
 The tractability is given in Theorem~\ref{thm:PlCSP_tractable:Hadamard}.
 When $\mathcal{F}$ is not a subset of $\mathscr{A}$, $\widehat{\mathscr{P}}$, or $\mathscr{M}$,
 we prove that $\PlHolant^c(\mathcal{F} \union \widehat{\EQ})$ is $\SHARPP$-hard,
 which is sufficient because of pinning (Theorem~\ref{thm:pinning}).
 
 For any degenerate signature $f \in \mathcal{F}$, we connect some number of $[1,0]$ to $f$ to get its corresponding unary signature.
 We replace $f$ by this unary signature, which does not change the complexity.
 Thus, assume that the only degenerate signatures in $\mathcal{F}$ are unary signatures.
 
 If $\mathcal{F} \not\subseteq \mathscr{A} \union \widehat{\mathscr{P}} \union \mathscr{M}$,
 then the problem is $\SHARPP$-hard by Theorem~\ref{thm:dichotomy:single_sig}.
 Otherwise, $\mathcal{F} \subseteq \mathscr{A} \union \widehat{\mathscr{P}} \union \mathscr{M}$ and we are done by Theorem~\ref{thm:mixing}.
\end{proof}

We also have the corresponding theorem for the $\PlCSP$ framework in the standard basis,
which is equivalent to Theorem~\ref{thm:PlCSP:words}.

\begin{theorem} \label{thm:PlCSP}
 Let $\mathcal{F}$ be any set of symmetric, complex-valued signatures in Boolean variables.
 Then $\PlCSP(\mathcal{F})$ is $\SHARPP$-hard unless
 $\mathcal{F} \subseteq \mathscr{A}$, $\mathcal{F} \subseteq \mathscr{P}$, or $\mathcal{F} \subseteq \widehat{\mathscr{M}}$,
 in which case the problem is in $\P$.
\end{theorem}

\section*{Acknowledgements}

Both authors were supported by NSF CCF-0914969 and NSF CCF-1217549.
A preliminary version of this paper appeared in~\cite{GW13}.
We are very grateful to Jin-Yi Cai for his support, guidance, and discussions.
We also thank him for his careful reading and insightful comments on a draft of this work as well as for the proof of Lemma~\ref{lem:2nd_condition_equivalence}.
We thank Peter B\"{u}rgisser, Leslie Ann Goldberg, Mark Jerrum, and Pascal Koiran
for the invitation to present this work at the Dagstuhl seminar on computational counting
as well as all those at the seminar for their interest.
We also thank the anonymous referees for their helpful comments.

\bibliographystyle{plain}
\bibliography{bib}

\appendix

\pagebreak

\section{Venn Diagram of the Tractable Signatures} \label{sec:venn_diagram}

This section contains a Venn diagram of the tractable Pl-\#CSP signature sets in the Hadamard basis.
Each signature may also take an arbitrary constant multiple from $\mathbb{C}$.
This figure is particularly useful in Section~\ref{sec:mixing}, where we consider the complexity of multiple signatures from different tractable sets.
The definition of each tractable signature set is given in Section~\ref{sec:preliminaries}.

For a signature $f$, the notation ``$f_{\ge k}$'' is short for ``$\arity(f) \ge k$''.
Notice that $\mathscr{M}\cap\widehat{\mathscr{P}}-\mathscr{A}$ is empty.

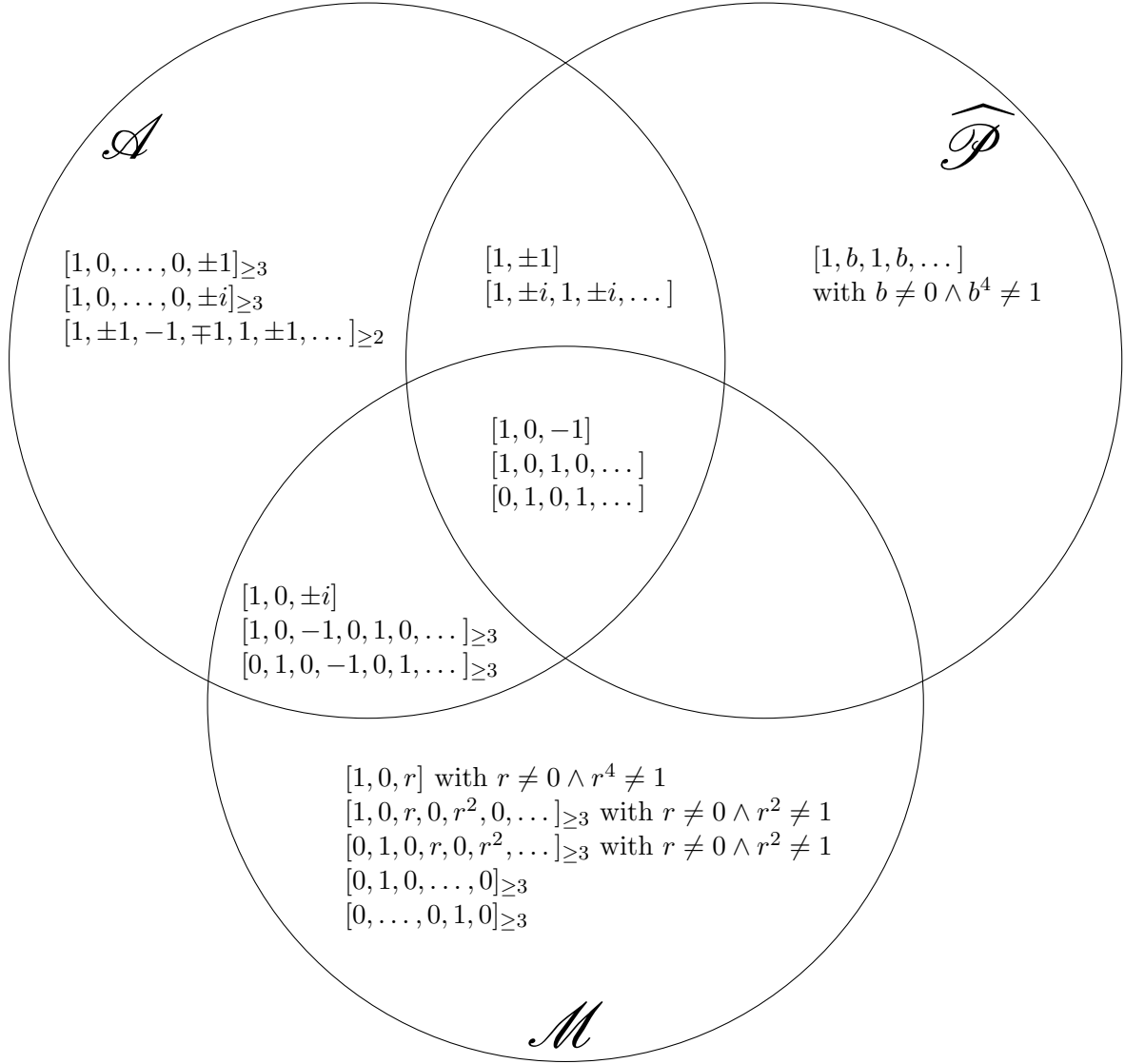
\begin{figure}[ht]
 \def\radius{5cm}
 \def\nodeDistVD{3.2cm}
 \begin{tikzpicture}
  \draw (150:\nodeDistVD) circle (\radius) node {};
  \draw  (30:\nodeDistVD) circle (\radius) node {};
  \draw (-90:\nodeDistVD) circle (\radius) node {};
  
  \draw node at (142:2.4 * \nodeDistVD) {\Huge $\mathscr{A}$};
  \draw node at  (40:2.3 * \nodeDistVD) {\Huge $\widehat{\mathscr{P}}$};
  \draw node at (-90:2.4 * \nodeDistVD) {\Huge $\mathscr{M}$};
  
  \draw node[text width=3.2cm] at (150:0.15 * \nodeDistVD)
   {\begin{itemize}[label=]
     \item $[1,0,-1]$
     \item $[1,0,1,0,\dotsc]$
     \item $[0,1,0,1,\dotsc]$
    \end{itemize}};
  \draw node[text width=3.7cm] at (95:0.9 * \nodeDistVD)
   {\begin{itemize}[label=]
     \item $[1, \pm 1]$
     \item $[1, \pm i, 1, \pm i,\dotsc]$
    \end{itemize}};
  \draw node[text width=6cm] at (220:1.02 * \nodeDistVD)
   {\begin{itemize}[label=]
     \item $[1, 0, \pm i]$
     \item $[1, 0, -1, 0, 1, 0, \dotsc]_{\ge 3}$
     \item $[0, 1, 0, -1, 0, 1, \dotsc]_{\ge 3}$
    \end{itemize}};
  \draw node[text width=7.1cm] at (150:1.6 * \nodeDistVD)
   {\begin{itemize}[label=]
     \item $[1, 0, \dotsc, 0, \pm 1]_{\ge 3}$
     \item $[1, 0, \dotsc, 0, \pm i]_{\ge 3}$
     \item $[1, \pm 1, -1, \mp 1, 1, \pm 1, \dotsc]_{\ge 2}$
    \end{itemize}};
  \draw node[text width=5cm] at (30:1.8 * \nodeDistVD)
   {\begin{itemize}[label=]
     \item $[1, b, 1, b, \dotsc]$
     \item with $b \ne 0 \wedge b^4 \ne 1$
    \end{itemize}};
  \draw node[text width=8.1cm] at (-90:1.6 * \nodeDistVD)
   {\begin{itemize}[label=]
     \item $[1, 0, r]$ with $r \ne 0 \wedge r^4 \ne 1$
     \item $[1,0,r,0,r^2,0,\dotsc]_{\ge 3}$ with $r \ne 0 \wedge r^2 \ne 1$
     \item $[0,1,0,r,0,r^2,\dotsc]_{\ge 3}$ with $r \ne 0 \wedge r^2 \ne 1$
     \item $[0,1,0,\dotsc,0]_{\ge 3}$
     \item $[0,\dotsc,0,1,0]_{\ge 3}$
    \end{itemize}};
 \end{tikzpicture}
 \caption{Venn diagram of the tractable Pl-\#CSP signature sets in the Hadamard basis.
  Each signature has been normalized for simplicity of presentation.
  For a signature $f$, the notation ``$f_{\ge k}$'' is short for ``$\arity(f) \ge k$''.
 }
 \label{fig:venn_diagram}
\end{figure}

\end{document}